\documentclass[11pt, twoside,final]{entics}

\usepackage{enticsmacro}
\usepackage{graphicx}

\usepackage{amssymb, amsmath}
\usepackage{microtype, stmaryrd, url, lmodern, eucal, extdash}
\usepackage{mathtools}
\usepackage{mathpartir}
\usepackage{braket}
\usepackage{todonotes}
\usepackage{quiver}
\usepackage[nameinlink]{cleveref}
\usepackage[inline,shortlabels]{enumitem}
\usepackage{tabularx}
\usepackage{circuitikz}

\usepackage{ccicons}

\makeatletter
\newcommand{\leqnomode}{\tagsleft@true}
\newcommand{\reqnomode}{\tagsleft@false}
\makeatother

\newenvironment{proofsketch}{%
\renewcommand{\proofname}{Proof sketch.}\proof}{\endproof}

\newcommand{\nm}{\mathsf} 
\newcommand{\bb}{\mathbb}

\renewcommand{\cal}{\mathcal}
\newcommand{\frk}{\mathfrak}

\renewcommand{\Set}{\mathsf{Set}}
\newcommand{\Frm}{\mathsf{Frm}}
\newcommand{\Loc}{\mathsf{Loc}}
\newcommand{\Top}{\mathsf{Top}}
\newcommand{\Retro}{\mathsf{Retro}}
\newcommand{\TopRetro}{\mathsf{TopRetro}}
\newcommand{\LocRetro}{\mathsf{LocRetro}}
\newcommand{\AmpLocRetro}{\mathsf{AmpLocRetro}}
\newcommand{\AmpTopRetro}{\mathsf{AmpTopRetro}}
\newcommand{\Comod}{\nm{Comod}}
\newcommand{\Mnd}{\nm{Mnd}}

\newcommand{\op}{\mathsf{op}}

\DeclareMathOperator{\ob}{\mathsf{ob}}
\DeclareMathOperator{\lan}{\mathsf{lan}}
\DeclareMathOperator{\pt}{\mathsf{pt}}
\newcommand{\inv}{^{-1}}
\newcommand{\id}{\mathsf{id}}

\newcommand{\codeno}[1]{\llparenthesis {#1} \rrparenthesis}
\newcommand{\deno}[1]{\llbracket {#1} \rrbracket}
\newcommand\bind{\ensuremath{\mathbin{>\mkern-6.8mu>\mkern-6.9mu=}}} 
\newcommand\rskip{\ensuremath{\mathbin{>\mkern-6.9mu>}}} 
\DeclareMathOperator{\return}{\nm{return}}

\newcommand\revepsilon{\reflectbox{$\varepsilon$}}

\makeatletter
\providecommand{\leftsquigarrow}{%
  \mathrel{\mathpalette\reflect@squig\relax}%
}
\newcommand{\reflect@squig}[2]{%
  \reflectbox{$\m@th#1\rightsquigarrow$}%
}
\makeatother

\def\conf{MFPS 2026}    
\volume{6}                     
\def\volu{6}                   
\def\papno{6}                  
\def\mydoi{19680}

\def\lastname{Garner, Renata, Wu}
\def\runtitle{Stone Duality for Monads}
\def\copynames{R. Garner, A. Renata, N. Wu}

\begin{document}
\begin{frontmatter}
  \title{Stone Duality for Monads}
  \author{Richard Garner\thanksref{macquarie}\thanksref{remail}}
  \author{Alyssa Renata\thanksref{imperial}\thanksref{anemail}}
  \author{Nicolas Wu \thanksref{imperial}\thanksref{anemail}}
  \address[imperial]{Imperial College London, London, UK}
  \address[macquarie]{Macquarie University, Sydney, Australia}
  \thanks[remail]{Email:{\texttt{\normalshape richard.garner@mq.edu.au}}}
  \thanks[anemail]{Email:{\texttt{\normalshape
  \{alyssa.renata19,n.wu\}@imperial.ac.uk}}}
  \begin{abstract}
    We introduce a contravariant idempotent adjunction between
    \begin{enumerate*}
    \item the category of ranked monads on $\Set$; and
    \item the category of internal categories and internal retrofunctors
      in the category of locales.
    \end{enumerate*} The left adjoint takes a monad $T$---viewed as a
    notion of computation, following Moggi---to its \emph{localic
    behaviour category} $\nm{LB}T$. This behaviour category is
    understood as ``the universal transition system'' for
    interacting with $T$: its ``objects'' are states and the
    ``morphisms'' are transitions. On the other hand, the right
    adjoint takes a localic category $\nm{LC}$---similarly understood
    as a transition system---to the monad $\Gamma\nm{LC}$ where
    $\Gamma\nm{LC}A$ is the set of $A$-indexed families of local
    sections to the source map which jointly partition the locale of
    objects. The fixed points of this adjunction consist of
    \begin{enumerate*}
    \item \emph{hyperaffine-unary monads}, i.e., those monads where
      term $t$ admits a read-only operation $\overline{t}$ predicting
      the output of $t$; and
    \item \emph{ample localic categories}, i.e., whose source maps are local
      homeomorphisms and whose locale of objects are strongly zero-dimensional.
    \end{enumerate*}
    The hyperaffine-unary monads arise in earlier works by Johnstone
    and Garner as a syntactic characterization of those monads with
    Cartesian closed Eilenberg-Moore categories. This equivalence is
    the \emph{Stone duality for monads}; so-called because it further restricts
    to the classical Stone duality by viewing a Boolean algebra $B$
    as a monad of $B$-partitions and the corresponding Stone space as
    a localic category with only identity morphisms.
  \end{abstract}
\end{frontmatter}

\section{Introduction}

Algebraic theories, as proposed by Plotkin and Power
\cite{Plotkin2002-uk,Plotkin2004-bt}, model notions of computation.
An algebraic theory consists of a set $\Sigma$ of generating computational
operations which define the program \emph{syntax}
along with some equations $\cal{E}$ that express program \emph{semantics}.
Expressing the semantics via equations means that we reason equationally about
programs, but there is another way to express semantics: two programs are
equivalent if they behave the same with respect to some external
\emph{environment}.
Reasoning against the environment can be more intuitive than
equational reasoning,
but constructing an explicit model of the environment from scratch is
a complicated
endeavor. Faced with this tension, in this paper we explore whether there is a
systematic way of recovering a model of the environment from the equations.

In logic, such explorations often manifest in the
construction of a \emph{duality}, which is a (contravariant) equivalence between
some category of algebraic objects (here the
algebraic theories) and some category of geometric objects (the environment).
The prototypical duality is the \emph{Stone duality} between Boolean
algebras and \emph{Stone spaces}. An element of a Boolean algebra $B$
is typically understood as a proposition about the environment, while
the elements of the corresponding Stone space $\nm{Spec}(B)$ are
the possible states of the environment. Now, the Boolean algebra is to be
recovered as certain \emph{decidable} functions on $\nm{Spec}(B)$, and so
this is formally enforced by imposing a topology on $\nm{Spec}(B)$ and taking
the continuous functions as a proxy for the decidable ones. The Stone spaces
characterize the spaces which may appear as the $\nm{Spec}(B)$ of some $B$, and
so $\nm{Spec}$ is an equivalence between the category of Boolean algebras and
Stone spaces.

Indeed, our exploration yields a duality. On the algebraic
side we have the (ranked) monads on $\Set$. In category theory,
monads are understood as
syntax-free algebraic theories \cite{Kelly1993-br}, because while every
algebraic theory induces a monad and vice versa, two algebraic theories with
different generating operations (i.e. different syntax) may induce the same
monad. We can replace algebraic theories by monads here because the
environment is sensitive only to the computations underlying the programs, and
not any particular choice of syntax. In subsections
\ref{subsec:monads},\ref{subsec:comodels} and \ref{subsec:terminal-comodel}, we
review the relevant aspects of monad theory.

On the geometric side, the environment associated to a monad $T$ is a category
$\nm{LB}T$ internal to ``a category of spaces $\cal{S}$'', which we call the
\emph{behaviour category} of $T$. Its objects are to be understood as
states of the environment and morphisms as transitions between states.
The internalization
in spaces achieves the same effect of controlling for computability as in the
classical Stone duality, but the first plot twist is that $\cal{S}$ is not the
category $\Top$ of topological spaces. Rather, it is the category $\Loc$
of \emph{locales}, also known as \emph{pointless
spaces} because they take as primitive the lattice of opens from the definition
of topological spaces while omitting the points entirely. The reason for
this twist is that there are monads whose topological space of states we found
to be empty (example \ref{ex:locale-of-injective-state-prelude}), and
yet whose localic space of states is non-trivial (example
\ref{ex:locale-of-injective-state})! Luckily, by
restraining ourselves to finitary monads it does
suffice to consider $\cal{S} = \Top$. The construction of the
localic/topological behaviour category is addressed in
\cref{sec:topological-localic-comodels} (for the object/state space)
and \cref{sec:localic-behaviour-category} (for the morphism/transition space).

Now, before we can state our duality result, we
must address the notion of morphism we take between our internal categories,
wherein lies the second twist. On the monad side, we take a standard
notion of monad morphism corresponding to interpretations of algebraic theories.
On the geometric side however, the notion of morphism we need is \emph{not}
functors, but rather \emph{retrofunctors}
\cite{NiuSpivak2025,Aguiar1997-ha,Ahman2016-fu,Ahman2017-gq}.
Roughly speaking, a retrofunctor $\bb{C} \rightsquigarrow \bb{D}$ takes objects
from $\bb{C}$ to $\bb{D}$, but takes morphisms from $\bb{D}$ to $\bb{C}$.
We can now state our first main contribution, though here
lies the final twist:
there is a contravariant \emph{adjunction}, and not equivalence(!), between the
category of
ranked monads on $\Set$ and the category $\LocRetro$ of
$\Loc$-internal categories and
retrofunctors, as well as its finitary cousin replacing
internalization in $\Loc$
by $\Top$.

\begin{center}
  \noindent
  \textbf{Theorem \ref{thm:stone-adjunction-loc}}
  $
  \begin{tikzcd}[ampersand replacement=\&,cramped]
    {\nm{LB} : \Mnd_r(\Set)} \& {\LocRetro^\op : \Gamma}
    \arrow[""{name=0, anchor=center, inner sep=0}, shift left=2,
    from=1-1, to=1-2]
    \arrow[""{name=1, anchor=center, inner sep=0}, shift left=2,
    from=1-2, to=1-1]
    \arrow["\dashv"{anchor=center, rotate=-90}, draw=none, from=0, to=1]
  \end{tikzcd}
  $ \hfill
  \textbf{Theorem \ref{thm:stone-adjunction-top}} $
  \begin{tikzcd}[ampersand replacement=\&,cramped]
    {\bb{B} : \Mnd_\omega(\Set)} \& {\TopRetro^\op : \Gamma_\omega}
    \arrow[""{name=0, anchor=center, inner sep=0}, shift left=2,
    from=1-1, to=1-2]
    \arrow[""{name=1, anchor=center, inner sep=0}, shift left=2,
    from=1-2, to=1-1]
    \arrow["\dashv"{anchor=center, rotate=-90}, draw=none, from=0, to=1]
  \end{tikzcd}
  $
\end{center}
These are merely adjunctions because on the
geometric side we assumed the environment is \emph{deterministic}, but any
algebraic theory with a commutative binary operation---such as the theory of
non-deterministic finite choice---cannot be modeled by a
deterministic environment
\cite{Uustalu2015-sv}, and so their corresponding environment under
our adjunction
is the trivial space, from which it is impossible to recover the
original theory.

In general, for a ranked monad $T$, the recovered back-and-forth monad
completes $T$ with \emph{prescience}: To each term
$t \in TA$, there is a new operation $\bar{t}$ which---intuitively---performs
$t$, keeps track of the result, and then rolls back the
state of the environment to just before performing $t$. Monads admitting such
\emph{prescient computations} were characterised by the first-named author in
\cite{Garner2024-yc,Garner2025-hr} as the \emph{cartesian closed}
monads, i.e., those whose categories of Eilenberg-Moore algebras are
cartesian closed. On the other side of the adjunction, the category we obtain
from a monad satisfies what we term \emph{ampleness}
(definition \ref{defn:ample-category}),
to use the terminology of
$C^\ast$-algebraists~\cite[Definition~2.2.4]{Paterson1999Groupoids}.
The cartesian closed monads and the ample localic
categories are precisely the fixpoints of our Stone adjunction, which
thus restricts
to an equivalence that constitutes the
\emph{Stone duality for monads} of the title, and our second main result.

\begin{center}
  \noindent
  \textbf{Theorem \ref{thm:stone-duality} }
  Adjunctions \ref{thm:stone-adjunction-loc} and \ref{thm:stone-adjunction-top}
  restrict to $\nm{HUMnd_r} \simeq \AmpLocRetro$ and
  $\nm{HUMnd_\omega \simeq \AmpTopRetro}$.
\end{center}

Now, the finitary version of the Stone duality obtains on the algebraic
side the finitary cartesian closed monads, and on the geometric side the ample
topological categories. From this, we
find that this finitary Stone duality subsumes the classical Stone
duality. on the algebraic side any Boolean algebra $B$ induces
a
finitary monad $T_B$ generated by binary operations of the
form $\textsf{if } b \textsf{ then } (-) \textsf{ else } (-)$ for each $b \in B$
\cite{Bergman1991-cc}. On the geometric side, the space of
objects for the category corresponding to $T_B$ is the classical Stone dual
$\nm{Spec}(B)$ of $B$, and the transitions consist of
only identity morphisms, reflecting the read-only nature of computations in
$T_B$. In fact, our story starts by inspecting the monad $T_B$ towards the
end of \cref{sec:preliminaries},
and curiously recovering $\nm{Spec}(B)$ via a construction on monads called
the \emph{terminal comodel}. This serves as a jumping off point for the rest of
the paper.

\textbf{Related and Future Work.}
The adjunction \ref{thm:stone-adjunction-top} is
related to existing work on \emph{interaction laws} \cite{Katsumata2020-bv}.
First, adjunction \ref{thm:stone-adjunction-top} admits a
detopologized version replacing
$\TopRetro$ by the category $\Retro$ of small categories and
retrofunctors. $\Retro$ is equivalent to the category of polynomial
comonads \cite{Ahman2016-fu,Ahman2017-gq}, and composing this
equivalence with our left adjoint functor $\bb{B}$, we recover the
\emph{cosemantics}
functor of Garner \cite{Garner2022-rj}. As explained in that paper,
the cosemantics comonad of a monad $T$ is its \emph{Sweedler dual} which is
characterized as the universal comonad admitting an \emph{interaction law} with
$T$ \cite{Katsumata2020-bv}. Here, the comonads represent environments
for interacting with the monad, so the detopologized $\bb{B}T$ represents
the universal environment for interacting with $T$. This brings us to our
first direction for future work: can we explain the topologized
$\bb{B}T$ or even
$\nm{LB}T$ as a universal interacting comonad? This seems to
demand a generalization of the Sweedler dual towards comonads on
$\Top$ or $\Loc$.

As we explained in the beginning, the behaviour category $\nm{LB}T$
should aid program reasoning. In future work, we hope to construct a logic for
reasoning about programs in Moggi's monadic metalanguage
\cite{Moggi1991-ff}. We expect this logic
to take the shape of propositional dynamic logic (PDL) \cite{Harel2001-aw}
since $\nm{LB}T$ can be seen as a Kripke model
whose propositions are interpreted as certain open sets of $\nm{LB}_0T$ and
whose programs are (generated by) $T1$. The modality $[m]\varphi$
is interpreted as $\codeno{m}\inv \varphi$. Since we are interpreting
the modalities in the monad $\Gamma\nm{LB}T$ completing $T$ with
propositions (the prescient computations), we don't require $T$ to contain
sufficiently many such propositions in the first place, lifting a constraint
in previous works on monadic program logic such as Goncharov \&
Schr\"oder \cite{Goncharov2013-ty}.

Having brought up Goncharov \& Schr\"oder \cite{Goncharov2013-ty},
we also ought to discuss their use of unary computations as
propositions, whereas our propositions are generated by binary computations
(see definition \ref{defn:behaviour-locale}).
This contrast seems to stem from their
assumption that computations may fail to terminate, leading to an
information ordering on computations \`a la domain theory, whereas we assume
that computations always terminate. Since this assumption
stems from a natural exploration of comodels, we hope in the future to
explore the Stone-type duality that arises when we consider comodels
\emph{residual} \cite{Ahman2020-bk,Katsumata2020-bv,Uustalu2020-ez}
over the lifting monad $\set{\bot} + -$. Just as the global
sections monad is a very fancy state monad, we expect the right
adjoint of this duality to take monads of partial sections, i.e.,
fancy partial state monads.

It is also natural to consider generalizations beyond monads on
$\Set$. In this paper, many constructions explicitly talk about
elements of monads, so an appropriate generalization will
have to talk about morphisms in the Kleisli category instead. We can
shed a preliminary light on this, based on the
adjunction introduced by Cockett and the first-named author
\cite{Cockett2021-jq} between
\emph{restriction categories realized in $\Loc$} and
\emph{$\Loc$-internal partite categories}. Here, $\nm{LB}T$ is (or rather,
generates) the partite internal category corresponding
under this adjunction to a restriction category generated by
the locale of states $\nm{LB}_0T$. If we can construct $\nm{LB}_0T$ in
element-free fashion, then we get a description of $\nm{LB}T$ which
avoids talking about elements of $T$ entirely.

Finally, the existence of a spectral duality for monads raise the
interesting possibility of developing a \emph{scheme theory of
monads}. Recall the notion of a \emph{scheme of rings} from
algebraic geometry: these are locally ringed sheaves which are
locally isomorphic to the spectrum of a commutative ring. The
analogy here is between rings and (hyperaffine-unary) monads, with
the spectrum of a ring (which is a sheaf) analogous to the source
map of $\nm{LB}T$. At first approximation, a scheme of monads
should then be a localic category $\nm{S}$, which ``locally
resembles'' $\nm{LB}T$ for some $T$. The idea is that in a scheme
$\nm{S}$, the monad (and hence the syntax) in play varies
continuously over the base space $\nm{S}_0$. This should allow for
the modelling of effects whose syntax is not fixed, such as
\emph{local} state---it would be nice if ``local state'' =
``locally a state monad''.

\section{Preliminaries, and the Classical Stone Duality from a
Monadic Perspective} \label{sec:preliminaries}

In subsections \ref{subsec:monads}, \ref{subsec:comodels},
\ref{subsec:terminal-comodel} and \ref{subsec:behaviour-category}, we
recall the basic theory of monads, their comodels, and the
classification of comodels by the behaviour category, which is the prototype
from which we build the topological and localic behaviour category.
We then review the classical Stone duality from a comodel-theoretic
perspective (subsection \ref{subsec:classical-stone-duality}),
using this to motivate our monadic Stone duality (subsection
\ref{subsec:generalized-stone-duality}).

\subsection{Monads} \label{subsec:monads}

Monads encode computations, as first
proposed by Moggi \cite{Moggi1991-ff}. We use a
definition which emphasizes this aspect of monads. Since this definition was
popularized in functional programming, we also we also borrow
some notation, specifically from the Haskell programming language.

\begin{definition}
  A \emph{monad} $(T,
  \bind, \return)$ on $\Set$ comprises:
  \begin{enumerate}
    \item for each set $A$, a set $TA$ of
      \emph{computations};
    \item for each \emph{value} $a \in A$, a \emph{pure
      computation} $\return a \in TA$; and
    \item for each set $A, B$ a
      \emph{composition} operation $\bind \colon TA \times TB^A \to TB$.
  \end{enumerate}
  These are required to satisfy, for each $t \in TA, a \in A, u \in
  TB^A, v \in TC^B$, the equations
  \[ \return a \bind u =
    u(a) \qquad  t \bind \lambda a. \return a = t \qquad (t \bind u) \bind v =
  t \bind (\lambda a. u(a) \bind v). \]
  A \emph{monad map} $\gamma \colon T \to S$
  comprises, for each set $A$, a function $\gamma_A \colon TA \to SA$
  satisfying $\gamma(\return a) = \return a$ and $\gamma(t \bind u) =
  \gamma(t) \bind \lambda a. \gamma(u(a))$. We also introduce an extra piece of
  notation: if $t_1 \in TA$ and $t_2 \in TB$, then we write $t_1 \rskip t_2$ for
  the composite $t_1 \bind \lambda _. t_2$ with the constant family on $t_2$.
\end{definition}

Here, the set $A$ of a computation $t \in TA$ is to be interpreted as
a return type.
A common way to specify monads is to specify a set $\Sigma$ of generating
operations $\nm{op}$ equipped
with their return types $A$, which we denote by $\nm{op}/A$ for brevity.
The monad $F_\Sigma$ generated by such a set of operations has, for
$F_\Sigma A$,
the set of trees whose leaves are values of $A$ and whose nodes are
the generating operations with the branching arity specified by the
return type. The operation $t \bind u$ substitutes each leaf of $t$
with label $a$ by the
tree $u(a)$, while $\return a$ is the leaf-only tree with label $a$.

One can furthermore specify generating equations, via a
family of binary relations
$\cal{E}_A$ over the set of trees $F_\Sigma A$. An \emph{algebraic
theory} $\bb{T}
= [\Sigma|\cal{E}]$ consists of a
set of generating operations $O$ and equations $R$. The monad $T$
generated by $\bb{T}$ is
given by the quotient $TA = F_\Sigma A/=_{\cal{E}_A}$ where the family of
equivalence relations $=_{\cal{E}_A}$
is generated by:
\begin{mathpar}
  \inferrule{t_1 ~\cal{E}_A~ t_2}{t_1 =_{\cal{E}_A} t_2} \quad
  \inferrule{t_1 =_{\cal{E}_A} t_2 \quad \forall a \in A. u_1(a) =_{R_B}
  u_2(a)}{t_1 \bind u =_{R_B} t_2 \bind u_2} \quad
  \inferrule{ t \in TA }{t =_{\cal{E}_A} t} \quad
  \inferrule{t_1 =_{\cal{E}_A} t_2}{t_2 =_{\cal{E}_A} t_1} \quad
  \inferrule{t_1 =_{\cal{E}_A} t_2 \quad t_2 =_{\cal{E}_A} t_3}{t_1
  =_{\cal{E}_A} t_3}
\end{mathpar}
If $=_{\cal{E}_A}$ is the identity relation for all $A$ then we say the
theory and its corresponding
monad is \emph{free}.

\begin{example}
  The monad of \emph{binary input} is the free monad corresponding to the theory
  $[\nm{flip}/2|\emptyset]$, whose generating operation $\nm{flip}$
  we can think of as sampling a binary digit, e.g.\! from a coin
  flip. A term $t \in TA$ of this monad is simply a binary tree with
  leaves labelled
  by elements of $A$. As a non-free example, the monad of
  \emph{reversible input} is induced by
  further adding two new unary operations $\nm{unflip_0/1}$ and
  $\nm{unflip_1/1}$
  along with the following equations expressing that the
  $\nm{unflip}$ operations
  are ``inverses'' of $\nm{flip}$.
  \[ \left[
      \begin{array}{c | c}
        \nm{flip}/2 & \nm{flip}(\nm{unflip}_0(a), \nm{unflip}_1(a)) =
        \return a \\
        \nm{unflip_0}/1 \quad \nm{unflip_1}/1 & \quad
        \nm{unflip}_0(\nm{flip}(a, b)) = \return a \qquad
        \nm{unflip}_1(\nm{flip}(a, b)) = \return b \qquad
      \end{array}
  \right] \]
\end{example}

For this paper, we will want to distinguish operations based on the
cardinality of their return types, and also consider monads whose operations
are bounded by a certain cardinality; these are the \emph{ranked} monads as
follows. A special case is when all operations have finite return type;
these are the \emph{finitary} monads.
After the following definition, any mention of monads will refer to
ranked monads only.

\begin{definition}
  An operation $t \in TA$ is
  \emph{finitary} if there is some function $f \colon I \to A$ from a
  finite set $I$ and $t' \in TI$ such that $t = t' \bind \lambda i.
  \return f(i)$; if here we replace ``finite'' by ``$\lambda$-small''
  for some regular cardinal $\lambda$, then we instead say that $t$ is
  \emph{$\lambda$-ary}. Now $T$ itself is \emph{finitary} if each of
  its operations is so, and is \emph{ranked} if there is a regular
  cardinal $\lambda$ for which all of its operations are $\lambda$-ary.
  We write $\Mnd_\omega(\Set)$ (resp. $\Mnd_r(\Set)$) for the category
  of finitary (resp. ranked) monads and monad maps.
\end{definition}

\subsection{Comodels} \label{subsec:comodels}

It is well-known that every monad is generated by some algebraic
theory, and hence
we can always interpret computations as trees. Such trees can be seen
as \emph{decision trees}, whose execution recursively proceeds by
executing the topmost operation against some \emph{environment} (See
  example
\ref{ex:comodels}  for a visualization). The environment
reciprocates by communicating which branch to walk down along, but is
itself changed
in the process. The execution then
proceeds by executing the subtree found along the given branch, and
so on. The execution
terminates when the walk
ends in a leaf, which becomes the return value of that execution.
Environments are mathematically given by the notion of
\emph{comodel}, as first introduced
in the context of computer science by Power and Shkaravska \cite{Power2004-qv}.
They are also called
\emph{stateful runners} by Uustalu \cite{Uustalu2015-sv}, and as
explained by Ahman and Bauer
\cite{Ahman2020-bk}, we may also think of comodels as virtual machines
for $\Sigma$.

\begin{definition} \label{defn:set-comodel}
  Let $T$ be a monad. A
  $T$-\emph{comodel} is a pair $(W, \codeno{-})$
  whose data comprises a set $W$ of \emph{states} along with
  \emph{co-interpretations} $\codeno{t} \colon W \to A \times W$ of
  each computation $t \in TA$. These co-interpretations are required to respect
  $\bind$ and $\return$, in the sense that for any $t \in TA$, $u : A
  \to TB$, $a \in A$
  and $w, w' \in W$,
  \[ \codeno{\return a}(w) = (a, w) \text{ and } \codeno{t \bind
  u}(w) = \codeno{u(a)}(w') \text{ where } (a, w') = \codeno{t}(w). \]
  A \emph{comodel map} is a function $W \to W'$ which preserves each
  co-interpretation,
  i.e., for each $t \in TA$, $\codeno{t}^{W'}(h(w)) =
  h(\codeno{t}^{W}(w))$. We write
  $\Comod_T$ for the category of $T$-comodels.
\end{definition}

\begin{example} \label{ex:comodels}
  An example comodel for the binary input monad is the set
  $2^\omega$ of infinite binary
  streams. Because co-interpretations respect $\bind$, the
  co-interpretations for a monad generated by a theory $[\Sigma |\cal{E}]$
  is determined by the co-intepretations of the generating
  operations $\Sigma$. In this case, it means that we need only specify a map
  $\codeno{\nm{flip}} : 2^\omega \to 2 \times 2^\omega$ for which we
  take the map $\beta \mapsto (\nm{head}(\beta), \nm{tail}(\beta))$.
  Then the co-interpretation of a term occurs by consuming one digit to walk
  down one operation at a time, as the
  righthand diagram demonstrates in co-interpreting $t =
  \nm{flip}(\nm{flip}(a,
  \nm{flip}(b, c)), d)$ against state
  $\beta = 0110\ldots$.
  The co-interpretation finishes when it reaches the leaf
  $c$, leaving the environment in state $\beta' = 0\ldots$, so
  $\codeno{t}(\beta) = (c, \beta')$.

  \begin{center}
    \resizebox{0.3\textwidth}{!}{%
      \begin{circuitikz}
        \tikzstyle{every node}=[font=\fontsize{18.2pt}{23.7pt}\selectfont]
        \draw [-{Latex[scale=1.5]}, ] (14.625,13.25) -- (14,12.5);
        \draw [-{Latex[scale=1.5]}, ] (14.125,12.25) -- (14.875,11.5);
        \node [font=\fontsize{10.2pt}{13.3pt}\selectfont, inner
        xsep=0.080cm, inner ysep=0.085cm, rounded corners=0.020cm] at
        (14.875,13.5) {$0110\ldots$};
        \node [font=\fontsize{10.2pt}{13.3pt}\selectfont, inner
        xsep=0.080cm, inner ysep=0.085cm, rounded corners=0.020cm] at
        (13.875,12.375) {$110\ldots$};
        \node [font=\fontsize{10.2pt}{13.3pt}\selectfont, inner
        xsep=0.080cm, inner ysep=0.085cm, rounded corners=0.020cm] at
        (15,11.375) {$10\ldots$};
        \node [font=\fontsize{10.2pt}{13.3pt}\selectfont, inner
        xsep=0.080cm, inner ysep=0.085cm, rounded corners=0.020cm] at
        (16,10.375) {$0\ldots$};
        \draw [-{Latex[scale=1.5]}, ] (15.25,11.125) -- (15.875,10.5);
        \node [font=\fontsize{18.2pt}{23.7pt}\selectfont, inner
        xsep=0.080cm, inner ysep=0.085cm, rounded corners=0.020cm] at
        (16.625,14.125) {$2^\omega$};
        \draw [line width=2pt, short] (10.875,13.375) -- (10.125,12.5);
        \draw [{Circle[scale=1.5]}-, short] (10.875,13.5) -- (11.625,12.625);
        \draw [{Circle[scale=1.5]}-, short] (10.125,12.5) -- (9.375,11.625);
        \draw [line width=2pt, short] (10.125,12.375) -- (10.875,11.5);
        \draw [line width=2pt, short] (10.75,11.625) -- (11.5,10.75);
        \draw [{Circle[scale=1.5]}-, short] (10.875,11.625) -- (10.125,10.75);
        \node [font=\fontsize{10.2pt}{13.3pt}\selectfont, inner
        xsep=0.080cm, inner ysep=0.085cm, rounded corners=0.020cm] at
        (9.375,11.375) {$a$};
        \node [font=\fontsize{10.2pt}{13.3pt}\selectfont, inner
        xsep=0.080cm, inner ysep=0.085cm, rounded corners=0.020cm] at
        (10.125,10.5) {$b$};
        \node [font=\fontsize{10.2pt}{13.3pt}\selectfont, inner
        xsep=0.080cm, inner ysep=0.085cm, rounded corners=0.020cm] at
        (11.625,10.5) {$c$};
        \node [font=\fontsize{10.2pt}{13.3pt}\selectfont, inner
        xsep=0.080cm, inner ysep=0.085cm, rounded corners=0.020cm] at
        (11.75,12.375) {$d$};
        \node [font=\fontsize{18.2pt}{23.7pt}\selectfont, inner
        xsep=0.080cm, inner ysep=0.085cm, rounded corners=0.000cm] at
        (10,14) {$t$};
        \draw  (15,11.75) circle (2.375cm);
      \end{circuitikz}
    }%
    \end{center}

  The same comodel works for reversible binary input, where the
  $\nm{unflip}$ operations
  put the corresponding digit back on top of the given stream, i.e.,
  $\codeno{\nm{unflip}_0} \colon \beta \mapsto 0:\beta$ and
  $\codeno{\nm{unflip}_1} \colon \beta \mapsto 1:\beta$.
  Clearly, any comodel of reversible input is also a comodel of
  binary input, but
  consider the set with two states which we shall call $(01)^*$ and
  $(10)^*$, intuitively the stream of repeating $01$s and $10$s
  respectively. This admits a cointerpretation
  $\codeno{\nm{flip}}((01)^*) = (0, (10)^*)$ and
  $\codeno{\nm{flip}}((10)^*) = (1, (01)^*)$, but this
  cannot be extended to the $\nm{unflip}$ operations in a way that
  would respect the generating equations, since this would require
  adding new states
  such as $0(01)^*$ to cointerpret $\codeno{\nm{unflip}_0}((01)^*)$.
\end{example}

\subsection{The Terminal Comodel} \label{subsec:terminal-comodel}

Hopefully, the reader has gathered enough intuition on comodels by
now to see that they look very much like transition systems. Indeed,
comodels are coalgebras, so it is instructive to look
at the terminal object which encodes the observable behaviours of
coalgebraic states.
In the case of
comodels, the observable behaviour of some state is to execute any
given computation
$t \in TA$ down to a return value $a \in A$, so the elements of
the terminal comodel
are given by maps $TA \to A$ (natural over $A$), admitting a certain
\emph{admissibility} condition \cite{Garner2022-rj}.

\begin{definition} (Admissible behaviours) \label{defn:terminal-comodel}
  A \emph{$T$-behaviour} $\beta$ simply consists of an assigment
  $\beta_A : TA \to A$ for each set $A$. Such a behaviour $\beta$ is
  \emph{admissible} if for any $t \in
  TA$, $a \in A$ and $u
  \colon A \to TB$, we have
  $\beta(t \bind u) = \beta(t \rskip u(\beta(t)))$ and $\beta(t
  \rskip \return a) = a$. Note that the latter equation is equivalent
  to naturality of $\beta : T \to \id_{\Set}$ in the presence of the
  former. The set of admissible
  behaviours $\bb{B}_0T$ admits a comodel structure with
  $\codeno{t}(\beta) = (\beta(t), \partial_t\beta)$ where the next
  behaviour $\partial_t \beta$ is defined by $\partial_t\beta \colon
  t' \mapsto \beta(t \rskip t')$. This makes $\bb{B}_0T$ the
  terminal comodel, where for any other comodel $W$, the
  unique comodel map $W \to \bb{B}_0T$ sends $w \in W$ to the
  behaviour $\beta_w \colon t \in TA \mapsto \pi_{A}(\codeno{t}(w))$.
\end{definition}

Unless otherwise indicated, every behaviour we consider from now on
will be admissible.

\begin{example}
  For our running example of the binary input monad, the comodel
  $2^\omega$ is in fact (isomorphic to) the terminal comodel. This is because
  by admissibility of any behaviour $\beta$, the
  execution of $\beta(t)$ for any tree $t$ is determined by the values of
  $\beta(\nm{flip}), \beta(\nm{flip} \rskip \nm{flip}), \ldots
  \beta(\nm{flip}^k) \ldots$, which determines a map $\omega \to 2$.
  For example,
  for the $t$ in example \ref{ex:comodels}
  $\beta(t)$ is determined by $\beta(\nm{flip}) = 0$,
  $\beta(\nm{flip}^2) = 1$ and
  $\beta(\nm{flip}^3) = 1$. Since the binary input monad is free, any such map
  $\omega \to 2$  determines an admissible behaviour, and hence we obtain our
  isomorphism. This argument also works for reversible binary input, since
  $\beta$ has no choice to make for computations involving the $\nm{unflip}$
  operations.
\end{example}

Any comodel $W$ admits a unique map $! \colon W \to \bb{B}_0T$, which by
taking inverse images determines a map $!\inv : \bb{B}_0T \to \Set$.
On the other hand,
not just any family $P : \bb{B}_0T \to \Set$ can induce a comodel,
but the first-named author
\cite{Garner2022-rj} showed that
we can equip $\bb{B}_0T$ with morphisms $\bb{B}_1T$ to obtain a
\emph{behaviour category} $\bb{B}T$ such
that covariant presheaf functors $P \colon \bb{B}T \to \Set$
correspond to comodels by taking the sum of the sets $P(\beta)$ over
all the behaviours $\beta$. We now review the construction of the
behaviour category.

\subsection{The Behaviour Category} \label{subsec:behaviour-category}
For a free monad $T$, a
$T$-behaviour $\beta$ produces, from each term $t \in TA$, a sequence
of operations called the \emph{trace}, which it had to evaluate to
run $t$ down to its return value. These traces are what we take as the morphisms
of the behaviour category. In fact, the return value itself
does not matter for computing the trace, i.e., $t \in TA$ and
$(t \rskip \return) \in T1$ have the same trace, so it is enough to consider
the trace of $T1$-terms only. Now, an arbitrary monad $T$ may have
more than one plausible set of generating operations, so we want a
definition of traces which is agnostic to any particular choice of generators.
We do this by defining trace-equivalence and then recover
the traces as the trace-equivalence classes.

\begin{definition} ($\beta$-equivalence, behaviour category)
  \label{defn:behaviour-category}
  Let $\beta$ be a $T$-behaviour. The relation $\sim_\beta$ on $T1$
  is the least equivalence relation such that $(t \bind u)
  \sim_{\beta} (t \rskip u(\beta(t)))$ for any $t \in TA$ and $u
  \colon A \to T1$. Write $[t]_\beta$ for the
  $\sim_\beta$-equivalence class of $t$. The \emph{behaviour
  category} $\bb{B}T$ has as objects $T$-behaviours, and as
  morphisms pairs of the form $(\beta, [t]_{\beta})$ but which we
  will simply write as $[t]_\beta$. The morphism
  $[t]_\beta$ has source $\beta$ and target $\partial_t\beta$. The
  identity $\id_\beta$ is $[\return]_\beta$, while the
  composite $[t]_\beta;[s]_{\partial_t\beta}$ is $[t \rskip s]_\beta$.
\end{definition}

\begin{theorem} \label{thm:comodels-classified-by-BT}
  $\Comod_T \simeq [\bb{B}T, \Set]$.
\end{theorem}

\begin{example} \label{ex:behaviour-category}
  Let $T$ be the monad of binary input.
  The trace of $t \in T1$ at $\beta$ is the number
  of $\nm{flip}$s traversed by $\beta$ when running down $t$, so a
  morphism in $\bb{B}T$ has the form $n \colon W \rightarrow
  \partial^n W$ for some $n \in \mathbb{N}$. 
\end{example}
\begin{example}
  For any set $S$, there is a \emph{state monad} $T = (- \times S)^S$ whose
  computations are functions which inspect and update the current state, along
  with using this information to produce a return value. Its behaviours are in
  bijection with the set of states $S$, and actually so are the traces
  $T1/_{\sim_s}$ over each $s \in S$. To
  understand why the latter is true, observe that the unary terms are
  just state-update functions $S \to S$.
  Now, intuitively, two such unary computations
  are equivalent if they make the same update at state $s$. Indeed,
  two unary terms $t_1, t_2 \in T(1) = S^S$ are $s$-equivalent iff $t_1(s)
  = t_2(s)$. So $\bb{B}T$ is the indiscrete (or chaotic) category
  with object-set $S$.
  We refer to \cite{Garner2022-rj} for more examples.
\end{example}

\subsection{The Classical Stone Duality, from a Comodel-theoretic Perspective}
\label{subsec:classical-stone-duality}

Recall that a Boolean algebra is a bounded distributive lattice
$(B, \wedge, \vee, \top, \bot)$ such that every element $b \in B$ admits a
complement, i.e., a (necessarily unique) element $\neg b$ such that $b
\vee \neg b = \top$
and $b \wedge \neg b = \bot$. From a logical perspective, we think of
the elements of $B$ as propositions about the world. Any Boolean
algebra induces a
\emph{monad of partitions} \cite{Bergman1991-cc,Garner2024-yc}, which we now
define ``out of thin-air''. However, this monad does admit an
equational presentation which we
actually use later in this paper (see remark
\ref{remark:equational-presentation-of-TB}).

\begin{definition}
  Let $B$ be a Boolean algebra. For a set $A$, an
  \emph{$A$-partition} of $B$ is a function
  $t : A \to B$ such that $p$ is
  \begin{enumerate*}
  \item \emph{finitely supported}, i.e., for all but finitely many
    $a \in A$, $t(a) = \bot$;
  \item \emph{pairwise disjoint}, i.e., $t(a) \wedge t(a') = \bot$
    for $a \neq a'$; and
  \item \emph{jointly covering},  i.e., $\bigvee \set{ t(a) \neq
    \bot | a \in A } = \top$.
  \end{enumerate*}
  The monad of partitions $T_B$ takes $T_{B}A$ to be the set of
  $A$-partitions, with
  $(\return a)(a) = \top$ and $\bot$ otherwise. The composite $t
  \bind u$ for $t \in TA$ and $u : A \to TC$ is the partition $(t
  \bind u)(c) = \bigvee_{a \in A} t(a) \wedge u(a)(c)$.
\end{definition}

While this monad seems quite complicated, it is generated by $T_B2$ which is
in fact isomorphic to the Boolean algebra $B$. To see why, notice
that a binary partition $t$
determines some $b := t(1)$, and then the conditions demand that
$t(0) = \neg b$. We can view such a $t$ as a program
expression $\textsf{if } b \textsf{ then } 1 \textsf{ else } 0$. In this case,
executing such an expression amounts to checking whether $b$ holds in
the world, and then returning the appropriate value. The observable
behaviour of the world then is just a specification of which propositions hold,
subject to some reasonable conditions. Such specifications are called
\emph{ultrafilters}.

\begin{proposition} \label{prop:terminal-comodel-of-partitions-monad}
  Let $B$ be a Boolean algebra. The terminal comodel $\bb{B}_0T_B$ is
  isomorphic to
  the set of \emph{ultrafilters} for $B$, namely those subsets
  $\frak{p} \subset B$ that satisfy:
  \begin{enumerate*}
  \item \emph{nontriviality}, i.e., $\top \in \frak{p}$;
  \item \emph{completeness}, i.e., $a \vee b \in \frak{p}$ implies
    $a \in \frak{p}$ or $b \in \frak{p}$; and
  \item \emph{consistency}, i.e., $\neg a \in \frak{p}$ implies $a
    \not\in \frak{p}$.
  \end{enumerate*}
\end{proposition}

Now, the much-celebrated Stone duality theorem says that we can actually recover
the Boolean algebra $B$ \emph{if} we additionally induce a topology on the set
of ultrafilters. The topological spaces that arise in this way are called the
\emph{Stone spaces} (definition \ref{defn:stone-spaces}), and by Stone duality
there is a dual equivalence $\nm{BA} \simeq \nm{Stone}^\op$.
In combination with the above proposition, the monads of the form $T_B$ are
therefore determined by their terminal comodel equipped with an
appropriate topology. Therefore it is natural to ask: can we generalize Stone
duality to recover any monad $T$ from its terminal comodel $\bb{B}_0T$?

\subsection{Towards a Generalized Stone Duality for Monads}
\label{subsec:generalized-stone-duality}

Since the terminal comodel interprets each computation $t$ as a map
$\codeno{t} \colon \bb{B}_0T \to A \times \bb{B}_0T$, it is natural
as a naive attempt to try and recover the monad $T$ as the
\emph{state monad} $(- \times \bb{B}_0T)^{\bb{B}_0T}$. But there are two
orthogonal problems to consider.

\emph{First}, we already know from
the classical Stone duality that we ought
to put some kind of topology on $\bb{B}_0T$, which we have yet to do.
Actually, there is a computational explanation for using topology: there
are too many functions $\bb{B}_0T \to A \times \bb{B}_0T$, not all of which
are actually computable. Consider the binary input or even
reversible input monad for example. There is no binary tree whose
cointerpretation is a map $2^\omega \to \set{\nm{yes}, \nm{no}}
\times 2^\omega$ assigning $\nm{repeat}(10) \mapsto \nm{yes}$,
and $\nm{no}$ for any other stream, because this requires inspecting
infinitely many digits of the stream. Putting a
topology, and asking for only continuous functions, acts as a
proxy for demanding computability\footnote{This is similar to
  Scott's model of the untyped lambda-calculus: the only set $X$
  isomorphic to $X^X$ is $X = 1$, but there are non-trivial
\emph{spaces} $X$ for which such isomorphisms exist.}. As we will
see later, the topology we put on $\bb{B}_0T$ coincides with the
cantor space topology when $T$ is binary input, which rules out
the function we just described. In
\cref{sec:topological-localic-comodels} we address this
topologization for arbitrary monads, capturing the Stone topology as
a special case. It turns out there that topology is only sufficient to capture
finitary monads, for reasons owing to size---hence our sensitivity to
cardinality
when introducing monads earlier. To better capture ranked monads, we need to
consider replacing topological spaces by an alternative notion of space called
\emph{locales}.

\emph{Second}, both the binary input monad and the
reversible input monad
induce the same terminal comodel, so we can't yet distinguish the two. Still,
they have different categories of comodels, as we saw
in example \ref{ex:comodels}. The solution for this is to
consider not just $\bb{B}_0T$, but the behaviour category
$\bb{B}T$; since it classifies categories of comodels, it must be able
to tell the difference between binary input and reversible input.
We have already seen in example \ref{ex:behaviour-category}
that the morphisms of the behaviour category for binary input are natural
numbers specifying how many digits of the stream to consume. For reversible
input, we have said the $\nm{unflip}$ operations are inverses for $\nm{flip}$,
and indeed the behaviour category for reversible input is a
\emph{groupoid} whose
morphisms are, in addition to the natural numbers, non-empty lists of binary
digits specifying digits to put back onto the stream. Amalgamating
our proposed solutions to the two problems, we must not
only topologize (or rather, localify)
$\bb{B}_0T$, but also $\bb{B}_1T$, and this is exactly what we
address in \cref{sec:localic-behaviour-category} to obtain
the \emph{localic behaviour category} $\nm{LB}T$ as the Stone dual of
$T$. The remaining two sections
\cref{sec:stone-adjunction-for-monads}
and
\cref{sec:stone-duality-for-haffun-monads}
take the story to its logical conclusion by constructing the dual equivalence
of categories that becomes our Stone duality for monads.

\section{Topological \& Localic Comodels}
\label{sec:topological-localic-comodels}

\subsection{Topological Comodels}

In this section, our goal is to topologize the terminal comodel. We do this
by rediscovering the terminal object in the category $\Comod_T(\Top)$ of
comodels valued in $\Top$, the category of topological spaces. For this reason,
we need to define comodels valued in other categories beyond $\Set$. To support
a definition of comodels, the base category $\cal{C}$ must have \emph{copowers},
which means that for any $C
\in \cal{C}$ and any set $A$, the $A$-fold coproduct $A \cdot C$
exists. For example, both $\Set$ and $\Top$  have copowers given by
the $A$-fold disjoint sum $\coprod_{a \in A} C$. We will denote the
inclusion maps by $\upsilon_a \colon C \to A \cdot C$, and the
codiagonal by $\pi_C \colon A \cdot C \to C$. 

\begin{definition} \label{defn:comodel} \textup{\cite[Definition
  2.15]{Garner2022-rj}}
  Let $T$ be a monad and $\cal{C}$ a category with copowers. A
  \emph{comodel} of $T$ in $\cal{C}$ is a pair $(W, \codeno{-})$
  whose data comprises an object $W\in \ob\cal{C}$ and
  \emph{co-interpretations} $\codeno{t} \colon W \to A \cdot W$ for
  each computation $t \in TA$. If we extend this to a
  cointerpretation $\codeno{u} \colon A \cdot W \rightarrow B
  \cdot W$ of each $u \colon A \rightarrow TB$ via $\codeno{u} :=
  [\codeno{u(a)}]_{a \in A}$, then the comodel axioms require that
  $\codeno{t \bind u} = \codeno{u} \circ \codeno{t}$ and
  $\codeno{\return a} = \upsilon_a$. A \emph{comodel map} $(W,
  \codeno{-}_W) \to (W', \codeno{-}_{W'})$ is a map $h \colon W \to
  W'$ which preserves each co-interpretation, i.e., for each $t \in
  TA$, we have $\codeno{t} \circ h = (A \cdot h) \circ \codeno{t}$.
  We write $\Comod_{T}(\cal{C})$ for the category of $T$-comodels in $\cal{C}$.
\end{definition}

This recovers definition \ref{defn:set-comodel} with $\Comod_T =
\Comod_T(\Set)$. Clearly, any $\Top$-comodel is a $\Set$-comodel,
yielding a forgetful
functor $U \colon \Comod_T(\Top) \rightarrow \Comod_T(\Set)$. On the
other hand, there is also a coarsest topology on any $\Set$-comodel
making it a $\Top$-comodel \cite{Garner2023-bi}:

\begin{definition} \label{defn:operational-topology} (Operational
  Topology) Let $W$ be a $T$-comodel in $\Set$. The \emph{operational
  topology} on $W$ is generated by sub-basic open sets $[t \mapsto a]
  := \set{w \in W | \codeno{t}(w) = (a, w') \text{ for some } w'}$
  for $t \in TA$ and $a \in A$.
\end{definition}

\begin{proposition} \label{prop:operational-topology-adjunction}
  The assignment which endows a comodel with its operational topology
  yields a right adjoint $O \colon \Comod_T(\Set) \rightarrow
  \Comod_T(\Top)$ to $U \colon \Comod_T(\Top) \rightarrow \Comod_T(\Set)$.
\end{proposition}
\begin{proof}
  Let $\cal{W}$ be a $\Top$-comodel and $W$ a $\Set$-comodel. Then
  any $\Set$-comodel morphism $f \colon U\cal{W} \to W$ is a
  continuous function $f \colon \cal{W} \to OW$ since $f\inv[t
  \mapsto a]_{W} = [t \mapsto a]_{\cal{W}} =
  \codeno{t}_{\cal{W}}\inv\left(\{a\} \times \cal{W}\right)$.
\end{proof}

Since right adjoints preserve limits---so in particular terminal
objects---we can
conclude that $\bb{B}_0T$ equipped with the operational topology is the terminal
$\Top$-comodel.

\begin{example}(The Cantor Space) \label{ex:cantor-space}
  Consider the monad of binary input and its terminal topological comodel. In
  this case, we can choose an even smaller set of sub-basic opens, namely those
  of the form $[\frk{b}] = \set{\beta | \beta \sqsupseteq \frk{b}}$ of those
  streams extending some finite sequence $\frk{b}$ of binary digits. This space
  is better known as the \emph{Cantor space}.
\end{example}

Now, if $T = T_B$ is the partitions monad over a Boolean algebra $B$,
since this monad is generated by the binary partitions in correspondence with
elements $b \in B$, the
operational topology on $\bb{B}_0T$ is generated by open sets of the form
$[b \mapsto 1]$ (and also $[b \mapsto 0]$, but those are recoverable
as $[\neg b \mapsto 1]$),
which we can simply write as $[b]$. Under the correspondence of
proposition \ref{prop:terminal-comodel-of-partitions-monad}, we find
that the generating open sets are of the form $[b] = \set{\frak{p} |
b \in \frak{p}}$, and
this is precisely the Stone topology on the set of ultrafilters. The spaces that
arise in this way are quite well-behaved, and are known as the
\emph{Stone spaces}.

\begin{definition} \label{defn:stone-spaces}
  A \emph{Stone space} $X$ is a topological space $X$ which is
  compact, Hausdorff
  and \emph{zero-dimensional}, i.e., admits a basis of clopen sets.
\end{definition}

For any monad $T$, $\bb{B}_0T$ is actually always Hausdorff and
by construction admits a basis of clopen sets (each $[t \mapsto a]$ is clopen),
but the compactness condition requires $T$ to be finitary, though we obtain
a proof of this as a corollary of propositions
\ref{prop:LB0T-ultraparacompact} and
\ref{prop:compact-and-ultraparacompact-is-stone}
from later in this section.

\begin{proposition} \label{prop:terminal-top-comodel-is-Stone}
  Let $T$ be a finitary monad. Then $\bb{B}_0T$ is a Stone space.
\end{proposition}

As for infinitary monads, the following example demonstrates that the terminal
topological comodel
for certain infinitary monads $T$ are poorly behaved, even if they admit a
deterministic and non-failing notion of environment.

\begin{example} \label{ex:locale-of-injective-state-prelude}
  Consider the monad $T$ corresponding to the following theory.
  \[\left[
      \begin{array}{c|c}
        \forall x \in \bb{R}, & \nm{get}_x \bind \lambda n.
        \nm{get}_x \lambda m. \return (n, m)
        = \nm{get}_x \lambda n. \return (n, n) \qquad \nm{get}_x \rskip
        \return a = \return a \\
        \nm{get}_x/\bb{N}     &  \nm{get}_x \bind \lambda
        n. \nm{get}_y \bind \lambda m. \return (n, m) = \nm{get}_y
        \bind \lambda m.
        \nm{get}_x \bind \lambda n. (n, m) \\
        &  \nm{get}_x \bind \lambda
        n. \nm{get}_y \bind \lambda m. \return (n, m) = \nm{get}_x
        \bind \lambda n.
        \nm{get}_y \bind \lambda m. f(n, m) \\
        & \text{ for all } x \neq y \in \bb{R} \text{ and } f \colon
        \bb{N} \times \bb{N} \to \bb{N} \times \bb{N} \text{
        satisfying } \forall n \neq m \in \bb{N}. f(n, m) = (n, m)
      \end{array}
  \right] \]
  The first three equations express that $\nm{get}_x$ fetches the value of a
  natural number from some memory location $x$, without changing the
  stored value,
  i.e., they express the theory of \emph{read-only state}. The behaviours
  correspond to functions $\bb{R} \to \bb{N}$ each of which express
  some memory configuration. The fourth equation is
  our non-standard addition, expressing that two distinct memory
  locations cannot contain
  the same value. This makes $\bb{B}_0T$ empty since the admissible
  behaviours in this case corresponds to injective memory
  configurations $\bb{R} \rightarrowtail \bb{N}$, of which there
  are (famously) none.

  However, a term in this signature of operations is well-founded,
  which is to say any
  particular execution of a term in $T$ can only query
  finitely many memory locations. So, from the program's
  perspective, it can never be sure that it is \emph{not} in a
  non-injective state configuration, since it needs to query
  uncountably-many memory locations to make a pigeonhole
  principle argument.
  We
  can think of this as having an uncountable virtual address space over
  countably many physical memory cells---the moment you try to
  query more locations than there are actual cells, the program
  forcibly halts.
\end{example}

Clearly, we cannot recover the monad $T$ of injective read-only state
from the empty
$\bb{B}_0T$, even though it is in some sense possible to reason
consistently about
such state configurations. In fact, there is a \emph{non-trivial
``space'' of injective functions $\bb{R} \rightarrowtail \bb{N}$}
\cite[Example C1.2.9]{Johnstone2002-xx}, but this
requires us to change what we mean by ``space'', choosing a notion of
space that emphasizes the opens, and de-emphasizes the points. Such
``spaces'' are called \emph{locales},
whose theory we now review, before computing the terminal localic comodel.

\subsection{Frames \& Locales}

Locales are spaces where we emphasize the opens, which in a topological space
forms a poset (under $\subset$) with infinite joins (unions
$\bigcup$) and finite meets (intersection $\cap$) that distribute.
Locales are defined more abstractly as any
such poset, but there is some subtlety on the morphisms.

\begin{definition}
  A \emph{frame} is a poset with infinite joins and finite meets
  satisfying the infinite distributive law $x \wedge \left(\bigvee_i
  y_i\right) = \bigvee_i (x \wedge y_i)$. We write $\Frm$ for the
  category of frames and frame homomorphisms, i.e., monotone maps
  which preserve infinite joins and finite meets. A \emph{locale} is
  simply a frame, but the category of locales is $\Loc := \Frm^\op$.
\end{definition}

The reason for this flipping of morphisms is that frame homomorphisms
more closely
resemble the inverse image map $f\inv$ associated to a continuous
function $f : X \to Y$, which
by continuity takes open sets of $Y$ to open sets of $X$, i.e.,
continuous functions
go in the opposite direction to their induced frame homomorphisms. In fact, we
have a functor $\cal{O} \colon \Top \to \Loc$ which sends $X$ to
its frame of open sets $\cal{O}(X)$, and a continuous map $f \colon
X \to Y$ to $f\inv \colon \cal{O}(Y) \rightarrow \cal{O}(X)$.
We imagine a general locale to be the lattice of opens of some space,
and the data of a continuous map to be given by the inverse image
map. To sustain this fantasy, we overload notation by also writing
$\cal{O} \colon \Loc^\op \to \Frm$ for the identity functor,
writing its action on morphisms as $\cal{O}(f) := f\inv$, and
calling elements $u \in \cal{O}(L)$ \emph{opens} of $L$.

\begin{definition}
  Any locale $L$ induces a topological space $\pt(L)$ with set of points
  the locale morphisms $x \colon 1 \to L$ from the terminal locale $1$
  and open sets given by $[u] = \set{x | x\inv u = \top}$ for $u \in
  \cal{O}(L)$. This yields a functor $\pt \colon \Loc \to \Top$,
  right adjoint to $\cal{O}$. A locale $L$ at which the adjunction
  counit is invertible is called \emph{spatial}, while a space at
  which the unit is invertible is called \emph{sober}: thus, spatial
  locales and sober spaces form equivalent categories.
\end{definition}

We make use of the fact---see \cite[II
2.11]{Johnstone1982-ei}---that frames may be constructed by
generators and relations.

\begin{example}(Copowers in $\Loc$)
  The (frame of opens of the)
  copower locale $A \cdot L$ can be presented by generating opens
  of the form $\braket{a \mapsto u}$ for $a \in A$ and $u \in
  \cal{O}(L)$, subject to the equations (1) $\bigvee_i \braket{a
  \mapsto u_i} = \braket{a \mapsto \bigvee_i u_i}$; (2)
  $\bigwedge^k_i \braket{a \mapsto u_i} = \braket{a \mapsto
  \bigwedge^k_i u_i}$; and (3) $\braket{a \mapsto u} \wedge
  \braket{a' \mapsto v} = \bot$ for $a \neq a'$. These equations
  imply that every open is uniquely of the form $\bigvee_{a \in A}
  \braket{a \mapsto u_a}$---which is consistent with the fact that
  $A$-fold copowers in $\Loc$ correspond to $A$-fold products in
  $\Frm$.
\end{example}

Since $\Loc$ has copowers, we may also consider
$\Comod_{T}(\Loc)$, and the adjunction between $\Top$ and $\Loc$ lifts to their
corresponding categories of $T$-comodels as the following proposition
establishes.

\begin{proposition} \label{prop:loc-top-adjunction-for-comodels}
  The adjunction $\cal{O} \dashv \pt \colon \Loc \rightarrow \Top$
  lifts to
  $\cal{O} \dashv \pt \colon \Comod_{T}(\Loc) \rightarrow \Comod_{T}(\Top)$.
\end{proposition}
\begin{proofsketch}
  Since comodels are defined in terms of copowers, it suffices to
  verify that $\cal{O}$ and $\pt$ preserve copowers. Since $\cal{O}$
  is a left adjoint, it preserves colimits and in particular
  copowers. As for $\pt$, this is one of
  its standard properties, which follows from fact that the
  terminal locale $1$ is connected, so that homming out of it
  preserves copowers.
\end{proofsketch}

\subsection{The Terminal Localic Comodel}
\label{subsec:terminal-localic-comodel}

By proposition \ref{prop:loc-top-adjunction-for-comodels}, we see
that the terminal topological comodel has to be the spatialization
of the terminal localic comodel. But by proposition
\ref{prop:operational-topology-adjunction}, the terminal
topological comodel is the set of behaviours, equipped with the
operational topology. This gives us an idea of what the terminal
localic comodel looks like.

\begin{definition} \label{defn:behaviour-locale}
  Let $T$ be a monad on $\Set$. The \emph{behaviour locale}
  $\nm{LB}_0T$ is generated by opens $[b]$ for each $b \in T2$,
  subject to the following equations for all $t \in TA, u :A \to
  TB, a \neq a' \in A$ and $b \in B$.
  \begin{center}
    \noindent
    \begin{minipage}{0.496\linewidth}
      \begin{equation}
        [t \mapsto a] \wedge [t \mapsto a'] = \bot
        \tag{$\nm{LB}_0$-$\bot$} \label{eqn:LB-bot}
      \end{equation}
    \end{minipage}
    \begin{minipage}{0.496\linewidth}
      \begin{equation}
        [t \rskip \return a \mapsto a] = \top
        \tag{$\nm{LB}_0$-$\eta$} \label{eqn:LB-return}
      \end{equation}
    \end{minipage}
  \end{center}
  \begin{equation}
    [t \bind u \mapsto b] = \bigvee_{a \in A} [t \mapsto a] \wedge
    [t \rskip u(a) \mapsto b] \tag{$\nm{LB}_0$-$\mu$} \label{eqn:LB-bind}
  \end{equation}
  Here we write $[t \mapsto a]$ as shorthand for $[t \bind \lambda
  a'. \delta_{a}(a')]$ with $\delta_{a}(a') = 1$ when $a = a'$ and
  $0$ otherwise.
\end{definition}

\begin{proposition} \label{prop:alt-axioms}
  Let $T$ be a monad on $\Set$. Then the following equations
  hold in $\nm{LB}_0T$:
  \begin{equation*}
    \bigvee_{a \in A} [t \mapsto a] = \top \qquad [t \rskip
    \return a \mapsto a'] = \bot \qquad [t \mapsto a] \wedge [t \bind
    u \mapsto b] = [t \mapsto a] \wedge [t \rskip u(a) \mapsto b],
  \end{equation*}
  where $t \in TA$, $a \neq a' \in A$, $u \colon A \to TB$ and $b \in B$.
\end{proposition}
In fact, axiom \eqref{eqn:LB-bind} can equivalently be replaced by
a combination of the first and third equations of proposition
\ref{prop:alt-axioms}. We chose
axioms $\eqref{eqn:LB-bind}$ and \eqref{eqn:LB-return} because of
their resemblance to the admissibility condition of $T$-behaviours
(definition \ref{defn:terminal-comodel}).
In any case, the axioms can be ``discovered'' as the necessary
conditions for proving the universal property of $\nm{LB}_0T$ as
the terminal localic comodel, as in the following proposition.

\begin{proposition} \label{prop:terminal-localic-comodel}
  $\nm{LB}_0T$ is the terminal localic comodel with
  cointerpretation $\codeno{t} \colon \nm{LB}_0T \to A \cdot
  \nm{LB}_0T$ given by: $\codeno{t}\inv \colon \braket{a_0 \mapsto
  [t_0]} \mapsto [t \mapsto a_0] \wedge [t  \rskip t_0]$.
\end{proposition}

Now by proposition \ref{prop:loc-top-adjunction-for-comodels} we
conclude that the spatialization of behaviour locale $\nm{LB}_0T$
is the topological space $\bb{B}_0T$. This says that $\nm{LB}_0T$
contains more information than $\bb{B}_0T$, so we ought to check that
it contains \emph{strictly} more by making sure that we have gained something
with respect to example
\ref{ex:locale-of-injective-state-prelude}.

\begin{example} \label{ex:locale-of-injective-state}
  Consider again the monad $T$ of injective read-only state from
  example \ref{ex:locale-of-injective-state-prelude}. Its
  behaviour locale $\nm{LB}_0T$, which in this instance is the the locale of
  injective functions $\bb{R} \rightarrowtail \bb{N}$ from \cite[Example
  C1.2.9]{Johnstone2002-xx}. The underlying frame of this locale is
  generated by opens $\braket{x \mapsto n}$ for $x \in
  \bb{R}$ and $n \in \bb{N}$, under the requirement that, for any $x \neq y$
  and $m \neq n$, the following equations holds:
  \[ \bigvee_{x \in \bb{R}} \braket{x \mapsto n} = \top \qquad
    \braket{x \mapsto n} \wedge \braket{x \mapsto m} = \bot \qquad
  \braket{x \mapsto n} \wedge \braket{y \mapsto n} = \bot. \]
  To see why, notice that by repeated application of axiom
  \eqref{eqn:LB-bind} the behaviour locale for this
  monad $\nm{LB}_0T$ can instead be
  generated by opens of the form $[\nm{get}_x \mapsto n]$, which of
  course coincide with
  the generating opens $\braket{x \mapsto n}$ above. Then it is a
  matter of proving that
  the equations of the behaviour locale reduce, in this instance, to
  the equations above---this proof is documented in
  \cref{subsec:proof-of-injective-state}.
\end{example}

\subsection{Strongly Zero-dimensional Locales and Grothendieck Boolean Algebras}

In proposition \ref{prop:terminal-top-comodel-is-Stone} the terminal
topological comodel is shown to be a Stone space whenever $T$ is
finitary. We now want to
state and prove a generalization of this proposition for infinitary
monads, dropping
the compactness requirement. In the absence of compactness, the
zero-dimensionality
condition of Stone spaces needs to be infinitarily strengthened,
obtaining the \emph{strongly
zero-dimensional} locales of Johnstone
\cite{Johnstone1997-ry}, also called \emph{ultraparacompact} by Van
Name \cite{Van-Name2013-tn}.

\begin{definition} \cite{Van-Name2013-tn}
  Let $L$ be a locale. An open $u \in \cal{O}(L)$ is
  \emph{complemented} if it so in the usual lattice-theoretic
  sense: thus, there is some $v \in \cal{O}(L)$ with $u \wedge v =
  \bot$ and $u \vee v = \top$. The set $\frk{B}(L)$ of complemented
  opens of $L$ inherits finite meets and joins from $L$, and so is
  a Boolean algebra. We say that $L$ is \emph{zero-dimensional} if
  $u \in \cal{O}(L)$ is the join of the complemented opens below it.

  A \emph{cover} of $L$ is a subset $J \subseteq \cal{O} (L)$ such
  that $\bigvee j = \top$. A cover $J$ \emph{refines} $J'$ if for
  every $u \in J$ there is $u' \in J'$ such that $u \leq u'$. An
  \emph{extended partition} $P$ is a pairwise disjoint cover, i.e.,
  $u \wedge v = \bot$ for any $u \neq v \in P$. A \emph{partition}
  $P$ is an extended partition which does not contain $\bot$, and any
  extended partition $P$ induces a partition $P^- = P
  \setminus\set{\bot}$. A zero-dimensional locale $L$ is
  \emph{strongly zero-dimensional} or \emph{ultraparacompact} if
  every open cover is refined by a partition.
\end{definition}

Since ultraparacompact locales generalize Stone spaces, there should be a
corresponding generalization of Stone duality. On the algebraic side, the
corresponding generalization is to consider the
\emph{Grothendieck Boolean algebras}. The notion is originally due
to~\cite{Van-Name2013-tn} but our nomenclature
follows~\cite[Definition 3.6]{Garner2024-yc}.

\begin{definition}
  A \emph{Grothendieck Boolean algebra} $B_\cal{J}$ is a Boolean
  algebra equipped with a \emph{strongly zero-dimensional
  topology}, i.e., a collection $\cal{J}$ of partitions for $B$ such that
  \begin{enumerate}[(i),itemsep=0em]
    \item $\cal{J}$ contains every finite partition;
    \item if $P \in \cal{J}$ and $Q_b \in \cal{J}$ for each $b \in
      P$, then $P;Q := \set{b \wedge c | b \in P, c \in Q_b}^- \in
      \cal{J}$ also;
    \item if $P \in \cal{J}$ and $f \colon P \to I$ is a surjective
      function, then each $\bigvee f\inv(i)$ exists in $B$ and
      $f\inv P := \set{\bigvee f\inv(i) | i \in I} \in \cal{J}$.
  \end{enumerate}
\end{definition}

\begin{theorem} \textup{\cite[Theorem 24]{Van-Name2013-tn}}
  \label{thm:ultraparacompact-duality}
  The category of ultraparacompact locales is dually equivalent to
  the category of Grothendieck Boolean algebras.
\end{theorem}
\begin{proof}
  We sketch just the constructions. An ultraparacompact locale $L$
  induces a Boolean algebra $\frk{B}(L)$ with strongly
  zero-dimensional topology given by the partitions of $L$ (the
  opens in a partition are necessarily complemented). On the other
  hand, a Grothendieck Boolean algebra $B_\cal{J}$ generates a
  locale of $\cal{J}$-closed ideals in the usual way (as
    explained by Vickers \cite{Vickers1996-vs} or Johnstone \cite[II
  2.11]{Johnstone1982-ei}).
\end{proof}

We are now in the position to show that $\nm{LB}_0T$ is
ultraparacompact. The intuition
here is that the equations defining $\nm{LB}_0T$ can be re-expressed
as generating equations for a Boolean algebra plus a strongly zero-dimensional
topology on that algebra.
The topology is essentially generated by the first equation of proposition
\ref{prop:alt-axioms}). When $T$ is
finitary, this is expressible as a finite disjunction, so the topology
is the topology of finite partitions, and hence $\nm{LB}_0T$ is compact because
every open cover in $\nm{LB}_0T$ is by construction refinable by a
finite partition,
which yields a finite subcover. Hence we have the following proposition, whose
proof details can be found in \cref{subsec:proof-of-LB0T-ultraparacompact}.

\begin{proposition} \label{prop:LB0T-ultraparacompact}
Let $T$ be a monad. Then $\nm{LB}_0T$ is ultraparacompact. Moreover,
if $T$ is finitary, then $\nm{LB}_0T$ is compact.
\end{proposition}

Now, we arrived at the notion of ultraparacompactness by eliminating
compactness,
so in the presence of compactness, ultraparacompactness amounts to
being a Stone space. So in particular, while example
\ref{ex:locale-of-injective-state} shows that we need
the full force of locales for infinitary
monads, it suffices to consider the Stone space $\bb{B}_0T$ for finitary monads.

\begin{proposition} \label{prop:compact-and-ultraparacompact-is-stone}
Let $L$ be a compact locale. Then $L$ is ultraparacompact iff $L$ is spatial and
$\pt L$ is a Stone space, so in particular $\pt \nm{LB}_0T \cong \bb{B}_0T$ is a
Stone space (thus also proving proposition
\ref{prop:terminal-top-comodel-is-Stone}).
\end{proposition}

\begin{remark}
Every Boolean algebra can be regarded as a Grothendieck Boolean
algebra under the topology of finite partitions. In that case,
the dual locale is spatial, and as a topological
space is precisely the classical Stone dual. In this
way, the equivalence of \ref{thm:ultraparacompact-duality} subsumes
the usual Stone duality.
\end{remark}

\section{The Localic Behaviour Category}
\label{sec:localic-behaviour-category}

In this section, we construct the localic behaviour category
$\nm{LB}T$. Following the construction of the behaviour category,
the locale of objects $\nm{LB}_0T$ can be characterized as the
terminal localic comodel. Hence, the main battle in this section is defining
the \emph{locale of morphisms} $\nm{LB}_1T$.

Consider the usual behaviour category to begin with.
For a monad $T$, any operation $t \in TA$ induces a section
$\eta(t) \colon \bb{B}_0T \to A \cdot \bb{B}_1T$ of the behaviour
category's source map $\sigma : \bb{B}_1T \to \bb{B}_0T$. This
suggests we should recover a monad from
the behaviour category by taking such sections. But as we have
already explained in subsection \ref{subsec:classical-stone-duality},
we must consider only the \emph{continuous sections}.

Now, we can flip this argument on its head: we are trying to construct a map
$\sigma : \nm{LB}_1T \to \nm{LB}_0T$ without knowing the domain
space. The only criterion
we have is that we know what continuous sections it should admit,
namely the ones
induced by $\eta(t)$ for all $t \in TA$. In this case, it is well-known that the
problem admits a universal solution: first, describe a sheaf
consisting of the ``sections''
which you want to be continuous, and then second, apply the
well-known correspondence
between sheaves and \emph{local homeomorphisms} to obtain the desired map
$\sigma : \nm{LB}_1T \to \nm{LB}_0T$.

In order to describe this sheaf, we need a way to talk about systems
for generating
sheaves by generators and relations. Luckily, our base space
$\nm{LB}_0T$ being strongly
zero-dimensional makes this quite a pleasant experience, because
sheaves over $\nm{LB}_0T$ are
``almost'' expressible as algebraic objects admitting an action by
the Grothendieck Boolean
algebra $B_{\cal{J}}$ generating $\nm{LB}_0T$. These algebraic
objects are called $B_{\cal{J}}$-sets, whose theory we now introduce.

\subsection{Sheaves, Local Homeomorphisms \&
\texorpdfstring{$B_{\cal{J}}$}{BJ}-Sets}

An important aspect of our results is that the source map $\sigma
\colon \nm{LB}_1T \rightarrow \nm{LB}_0T$ of the localic behaviour
category is a local homeomorphism. Here, a map $f$ of locales is a
\emph{local homeomorphism} if there is a cover $\set{v_i}_i$ of its
domain such that, on each part of this cover, the map $f$ restricts
to an open injection. 
It is well-known that local homeomorphisms into a locale correspond
to sheaves on a locale; and since $\nm{LB}_0T$ is ultraparacompact,
this leads to a particularly appealing description of the source
map. For indeed, sheaves on an ultraparacompact locale---or at
least, those possessing a global section---can be described purely
algebraically as sets with a suitable action of the corresponding
Grothendieck Boolean algebra. This was first shown by Bergman
\cite{Bergman1991-cc} for Boolean algebras, and later extended to
the Grothendieck case \cite{Garner2024-yc}. One advantage of
this presentation is that it makes clear what homomorphisms and
congruences of $B_{\cal{J}}$-sets are.

\begin{definition} ($B_{\cal{J}}$-sets) \label{defn:BJ-set}
Let $B$ be a non-degenerate Boolean algebra (i.e., $0 \neq 1$ in
$B$). A $B$-set $F$ consists of a set $|F|$ equipped with one
binary operation $b(-, -)$ for each $b \in B$ satisfying the equations
\begin{equation}
  \begin{alignedat}{3}
    & b(x, x) = x \qquad && b(b(x, y), z) = b(x, z) \qquad &&
    b(x, b(y, z)) = b(x, z) \\
    & \top(x, y) = x \qquad && (\neg b)(x, y) = b(y, x) \qquad &&
    (b \wedge c)(x, y) = b(c(x, y), y)
  \end{alignedat}  \label{eqn:B-set-axioms}
\end{equation}
for all $x, y, z \in |F|$. If $\cal{J}$ is a strongly
zero-dimensional topology on $B$, then a $B_\cal{J}$-set $F$
consists of a $B$-set $F$ further equipped with a $P$-ary
operation $P \colon |F|^P \to |F|$ for each partition $P \in
\cal{J}$. These operations are required to satisfy, for any $z
\in |F|$ and families $x, y \in |F|^P$, the axioms
\begin{equation}
  P(\lambda b. z) = z \qquad P(\lambda b. b(x_b, y_b)) =
  P(\lambda b. x_b) \qquad \text{and} \qquad b(P(x), x_b) =
  x_b\rlap{ .} \label{eqn:BJ-set-axioms}
\end{equation}
\end{definition}

These axioms are rather intuitive if one reads each operation $b$
as an if-then-else operation, and the infinitary operations $P$ as
infinitary switch statements. To see the correspondence with
sheaves, view the elements of a $B_{\cal{J}}$-set as a global
section. Then the operations perform amalgamation: for example
$b(x, y)$ is the unique amalgam of $x|_{b}$ and $y|_{\neg b}$. We
don't have to explicitly track local sections because if
we have any global section $t$ at all, then a local section $s$
over $b$ manifests as a global section by taking the amalgamation
of $s$ and $t|_{\neg b}$. Hence, the category of non-empty
$B_{\cal{J}}$-sets is equivalent to the category of sheaves over
the locale presented by $B_{\cal{J}}$ that have a global section.

Because every local section of a $B_{\cal{J}}$-set $F$ comes from
some global section, the set of local sections over some $b$ is a
quotient of $|F|$, by the relation $\equiv_b$ defined as follows.

\begin{proposition} \cite[Proposition 2.6]{Garner2025-hr}
Let $B_\cal{J}$ be a non-degenerate Grothendieck Boolean algebra.
Any $B_\cal{J}$-set structure on a set $|F|$ induces equivalence
relations $\equiv_b$ for $b \in B$ given by $x \equiv_b y \iff b(x,
y) = y$. These equivalence relations satisfy:
\begin{enumerate*}
\item if $x \equiv_b y$ and $c \leq b$ then $x \equiv_c y$;
\item $x \equiv_\top y$ iff $x = y$, and $x \equiv_\bot y$ always;
\item for any $P \in \cal{J}$ and $x \in X^P$, there is a unique
  $z \in X$ such that $z \equiv_b x_b$ for all $b \in P$.
\end{enumerate*}
In fact, any $B$-indexed family of equivalence relations on $|F|$
satisfying (i)--(iii) determine a $B_\cal{J}$-set, wherein
$P(\lambda b.x_b)$ is the aforementioned unique $z$. With this
alternative definition, a $B_\cal{J}$-set homomorphism is a
function that preserves the $\equiv_b$ relations.
\end{proposition}

The sheaf corresponding to the source map will be constructed as a
quotient of a free $B_{\cal{J}}$-set, so it is instructive to
construct the free $B_{\cal{J}}$-set explicitly.

\begin{definition} (Free $B_{\cal{J}}$-sets) \label{defn:free-BJ-sets}
Let $A$ be a set and $B_{\cal{J}}$ a non-degenerate Grothendieck
Boolean algebra. Then the \emph{Grothendieck Boolean power}
$A[B]^{\cal{J}}$ is the set of functions $h \colon A \to B$ for
which $\set{h(a) | a \in A}^- \in \cal{J}$.
\end{definition}

\begin{remark} \label{remark:equational-presentation-of-TB}
Notice that the following
construction is precisely the monad of partitions $T_B$ when taking $\cal{J}$
to be the topology of finite partitions. In other words, the $B$-sets are
precisely the Eilenberg-Moore algebras of $T_B$, and definition
\ref{defn:BJ-set}
is the equational presentation of $T_B$. Of course, we can equally well
define a monad of partitions $T_{B_{\cal{J}}}$ for Grothendieck
Boolean algebras.
\end{remark}

\begin{proposition} \cite[Remark 3.17]{Garner2024-yc}
\label{prop:free-BJ-sets}
Let $A$ be a set. Then $A[B]^{\cal{J}}$ has a $B_{\cal{J}}$-set
structure given by $P(\lambda b. h_b) := \lambda a. \bigvee_b b
\wedge h_b(a)$, and this is the free $B_{\cal{J}}$-set with
$A$-many generators. The unit map $A \to A[B]^{\cal{J}}$
identifies $a \in A$ with the map $\delta_a$ for which
$\delta_a(a) := \top$ and $\delta_a(a') := \bot$ for $a' \neq a$.
\end{proposition}

Given a sheaf $F$ on a locale $L$, the corresponding local
homeomorphism $E(F) \to L$ is found by taking $\cal{O}(E(F))$ as
the frame of subsheaves of $F$. We can re-express this in terms of
the category of sheaves on $L$: this is a topos, and in
particular admits a subobject classifier $\Omega$, so
subsheaves of $F$ correspond to sheaf maps $F \to \Omega$. Now if
$L$ is the ultraparacompact locale presented by $B_{\cal{J}}$, then
$\Omega$ itself is a $B_{\cal{J}}$-set, and so the frame
$\cal{O}(E(F))$ is given by the set of $B_{\cal{J}}$-set
homomorphisms $F \rightarrow \Omega$ under pointwise ordering.
$\Omega$ as a $B_\cal{J}$-set turns out to be $B_\cal{J}$ itself
with the action $P(\lambda b. u_b) = \bigvee_{b \in P} (u_b \wedge
b)$, or equivalently with $u \equiv_b v \iff b \wedge u = b \wedge v$.

\begin{definition} \label{defn:etale-space-of-BJ-set}
Let $L$ be an ultraparacompact locale presented by $B_\cal{J}$,
and $F$ be a $B_{\cal{J}}$-set. The \emph{\'etale space} $E(F)$
corresponding to $F$ is given by $\cal{O}(E(F)) :=
\nm{Set}_{B_{\cal{J}}}(F, \cal{O}L)$, and its associated
\emph{projection map} $\sigma \colon E(F) \to L$ is defined by
$\sigma\inv \colon u \mapsto \nm{const}_u$ (the constant function at $u$).
\end{definition}

\begin{lemma} \label{lemma:basis-of-etale-space}
Let $L$ be an ultraparacompact locale presented by $B_\cal{J}$,
and $F$ be a $B_{\cal{J}}$-set. Each element $x \in |F|$
injectively induces an open $\hat{x} \in \cal{O}(E(F))$ defined
by $\hat{x} := \lambda y. \bigvee \set{b | x \equiv_b y}$.
Moreover, these generate $\cal{O}E(F)$ because every $w \in
\cal{O}(E(F))$ can be expressed as $w = \bigvee_{x \in |F|}
\hat{x} \wedge \nm{const}_{w(x)}$.
\end{lemma}

\begin{proposition} \label{prop:corresponding-local-homeo}
Let $L$ be an ultraparacompact locale presented by $B_\cal{J}$.
Then the corresponding local homeomorphism of a $B_\cal{J}$-set
$F$ is the map $\sigma \colon E(F) \to L$.
\end{proposition}

For a sheaf $F$, the points of the corresponding local
homeomorphism $E(F)$ are known as \emph{germs}. Here is the
corresponding notion for $B_\cal{J}$-sets.

\begin{proposition} \label{prop:germs-of-local-homeo}
Let $L$ be an ultraparacompact locale presented by $B_\cal{J}$
and $F$ a $B_\cal{J}$-set. Then $\pt E(F) \cong \sum_{p \in \pt
L} |F|/_{\equiv_p}$, where $x \equiv_p y \iff \exists b \ni p. x
\equiv_b y$. An element $[x]_p$ of this space is called a
\emph{germ}. The topology on this space is generated by subbasic
open sets $[x|b] = \set{[x]_p | p \in b}$.
\end{proposition}

\subsection{The Locale of Morphisms}

Let $B_{\cal{J}}$ denote the Grothendieck Boolean algebra of complemented opens
in $\cal{O}\nm{LB}_0T$. We can now construct the $B_{\cal{J}}$-set which will
eventually become our locale of morphisms $\nm{LB}_1T$. We know that this
$B_{\cal{J}}$-set should be generated by computations in $T$, but it
should not be generated
\emph{freely}: we can see in $\bb{B}T$ that if a term factors as $t
\bind u$, then at each $\beta \in [t \mapsto a]$ we have $t \bind u
\sim_\beta t \rskip u(a)$, so the global sections $\eta(t \bind u)$ and
$\eta(t \rskip u(a))$ are equal when restricted to the region $[t \mapsto a]$.
Here the advantage of the $B_{\cal{J}}$-set approach
becomes clear: the condition relating $t \bind u$ with $t \rskip
u(a)$ above can be expressed purely equationally, as the following congruence
relation on the $B_{\cal{J}}$-set freely generated by $T1$.

\begin{definition} \label{defn:sheaf-of-transitions}
Let $T$ be a monad. The \emph{sheaf of transitions} $F_T$
associated to $T$ is a quotient of the free $B_\cal{J}$-set
$T1[B]^{\cal{J}}$ with generators $T1$, by the smallest
$B_\cal{J}$-congruence identifying
$t \bind u \approx P^{(t)}(\lambda [t \mapsto a]. t \rskip u(a)) $
where $t \in TA$, $u \colon A \to T1$ and $P^{(t)} = \set{[t
\mapsto a] | a \in A}^-$ is the partition canonically associated to $T$.
\end{definition}

By our previous considerations, we expect the local homeomorphism
constructed from $F_T$ to be
the locale of morphisms for our localic behaviour category. But it is
not at all obvious what its relationship is with the morphisms of the behaviour
category introduced in definition \ref{defn:behaviour-category}. To
see the connection, we prove that two elements $x, y \in
T1[B]^{\cal{J}}$ are related by $\approx$ just when they are
pointwise trace-equivalent, as expressed by the following
definition and accompanying lemma.

\begin{definition}
Let $T$ be a monad of rank $\kappa$. Define the
$\nm{LB}_0T$-valued relation of \emph{trace equivalence} on $T1$:
\begin{alignat}{2}
  \deno{m \sim m'} & = \bigvee_{k \geq 1} \deno{m \sim_k m'}
  \quad \text{ where } \quad \deno{m_1 \sim_k m_k} = \bigvee
  \set{\textstyle{\bigwedge^{k-1}_{i = 1}} \deno{m_i \sim_1
  m_{i+1}} | m_2 \ldots m_{k - 1} \in T1} \\
  \deno{m \sim_1 m'} & = \bigvee \left\{ [t \mapsto a] ~\middle|~
    \begin{matrix}
      A \in \Set, |A| \leq \kappa, t : TA, u, u' \colon A \to T1,
      a \in A, \\
      u(a) = u'(a), m = t \bind u, m' = t \bind u'
  \end{matrix}\right\}
\end{alignat}
\end{definition}
If $u \in \cal{O}\nm{LB}_0T$ is a complemented open, we define $m
\sim_u m'$ to be true whenever $u \leq \deno{m \sim m'}$ (if $u =
\top$, we simply write $m \sim m')$. This definition seems
complicated, but it is just the point-free transliteration of the
definition for $\beta$-equivalence. As $\beta$-equivalence is an
equivalence relation, so is $\deno{- \sim -}$, as the following lemma shows.

\begin{lemma}
$\deno{- \sim -}$ is an equivalence relation, in the sense that
for all $m_1, m_2, m_3 \in T1$,
\[ \deno{m_1 \sim m_1} = \top \qquad \deno{m_1 \sim m_2} \leq
  \deno{m_2 \sim m_1} \qquad \deno{m_1 \sim m_2} \wedge \deno{m_2
\sim m_3} \leq \deno{m_1 \sim m_3}. \]
Consequentially, all the $\sim_u$ are equivalence relations in
the usual sense.
\end{lemma}

Now, we come to the central lemma relating $\approx$ to the more familiar
$\beta$-equivalence. The proof of this lemma is quite involved, which
we relegate
to \cref{subsec:proof-of-pointwise-trace-equivalence}.

\begin{lemma} \label{lemma:pointwise-trace-equivalence}
Let $x, y \in T1[B]^{\cal{J}}$. Then $x \approx y$ iff for each
$m, n \in T1$, we have $m \sim_{x(m) \wedge y(n)} n$. Moreover,
$m \equiv_b n \iff m \sim_b n$.
\end{lemma}

Be warned that we abuse notation by confusing elements of $T1$ with
the induced element of $F_T$, even though the unit map $\delta
\colon T1 \to |F_T|$ from proposition \ref{prop:free-BJ-sets} is in
general not injective.

Recall from proposition \ref{prop:corresponding-local-homeo} that
the corresponding local homeomorphism is the locale of
homomorphisms $\nm{Set}_{B_{\cal{J}}}(F_T, \nm{LB}_0T)$. By the
universal property of free algebras, such homomorphisms correspond
to functions $w \colon T1 \to \nm{LB}_0T$ respecting the generating
equation $w(t \bind u) = P^{(t)}(\lambda [t \mapsto a]. t \rskip
u(a))$. We can restate this in terms of trace equivalence between
generators, as follows (proof of correspondence can be found in
\cref{subsec:proof-of-locale-of-transitions}).

\begin{definition} \label{defn:locale-of-transitions}
The \emph{locale of transitions} $\nm{LB}_1T$ is the
pointwise-ordered poset of functions $w \colon T1 \to \cal{O}(\nm{LB}_0T)$
for which $m_1 \sim_b m_2$ implies $w(m_1) \equiv_b w(m_2)$ for
any $m_1, m_2 \in T1$ and $b \in B$.
\end{definition}

We are now in the position to introduce the localic behaviour
category, but before we do so we specialize lemma
\ref{lemma:basis-of-etale-space} and proposition
\ref{prop:germs-of-local-homeo} to our locale of transitions, which
relates the localic behaviour category back to the usual behaviour category.

\begin{lemma} \label{lemma:basis-of-LB1T}
Every open $w \in \cal{O}(\nm{LB}_1T)$ can be expressed as $w =
\bigvee_{m} \hat{m} \wedge \nm{const}_{w(m)}$, so the frame
$\cal{O}(\nm{LB}_1T)$ is generated by opens of the form $[m | b]
:= \lambda n. \deno{m \sim n} \wedge [b]$ for $m \in T1$ and $b \in T2$.
\end{lemma}

\begin{proposition} \label{prop:germs-of-LB1T}
The set of points $\pt(\nm{LB}_1T)$ is bijective with $\bb{B}_1T$.
\end{proposition}

Finally, the following lemma is useful for we will often
have to consider various pullbacks with the source map, such as
when we define the composition map of the localic
behaviour category in definition \ref{defn:behaviour-category} below.

\begin{lemma} \label{lemma:pullback-with-LB1T}
The pullback $L \times_{\nm{LB}_0T} \nm{LB}_1T$ of a locale map
$f \colon L \to \nm{LB}_0T$ along the source map $\sigma \colon
\nm{LB}_1T \to \nm{LB}_0T$ has frame of opens given by the
pointwise-ordered poset of functions $h \colon T1 \to \cal{O}L$
for which $m_1 \sim_b m_2$ implies $h(m_1) \wedge f\inv b =
h(m_2) \wedge f\inv b$ for any $m_1, m_2 \in T1$ and $b \in B$.
\end{lemma}

In terms of points, such a function $h$ contains all the points
$(x, [m]_\beta)$ for which $x \in h(m)$ and $f(x) = \beta$.

\begin{definition} \label{defn:localic-behaviour-category}
Let $T$ be a monad. Then the \emph{localic behaviour
category} $\nm{LB}T$ has:
\begin{itemize}
  \item locale of objects $\nm{LB}_0T$ given by the terminal
    localic $T$-comodel;
  \item source map $\sigma \colon \nm{LB}_1T \to \nm{LB}_0T$
    given by $\sigma\inv(u) := \nm{const}_u$;
  \item target map $\tau \colon \nm{LB}_1T \to \nm{LB}_0T$ given by
    $\tau\inv(u) := \lambda m. \codeno{m}\inv u$;
  \item identity map $\iota \colon \nm{LB}_0T \to \nm{LB}_1T$ given by
    $\iota\inv(w) := w(\return)$;
  \item composition map $\mu \colon \nm{LB}_1T
    \times_{\nm{LB}_0T} \nm{LB}_1T \to \nm{LB}_1T$ given by
    $\mu\inv \colon w \mapsto \lambda m, n. w(m \rskip n)$,
    where, by lemma \ref{lemma:pullback-with-LB1T} we identify
    $\cal{O}(\nm{LB}_1T \times_{\nm{LB}_0T} \nm{LB}_1T)$ with the
    poset of functions $h \colon T1 \times T1 \to
    \cal{O}(\nm{LB}_0T)$ for which $m_1 \sim_b m_2$ implies $h(m_1,
    n) \equiv_b h(m_2, n)$ and $n_1 \sim_b n_2$ implies $h(m,
    n_1) \equiv_{\codeno{m}\inv b} h(m, n_2)$.
\end{itemize}
\end{definition}

A function $h \colon T1 \times T1 \to
\cal{O}(\nm{LB}_0T)$ as above corresponds to the open set
containing pairs of germs
$([m]_\beta, [n]_{\partial_m \beta})$ with $\beta \in h(m, n)$.
For the verification that this is a localic category, see
\cref{subsec:proof-of-localic-behaviour-category}.

\subsection{Topological Behaviour Category \& Finitary Monads}

We have already seen that $\pt \nm{LB}_0T \cong \bb{B}_0T$
(proposition \ref{prop:compact-and-ultraparacompact-is-stone}) and
$\pt \nm{LB}_1T \cong \bb{B}_1T$ (proposition
\ref{prop:germs-of-LB1T}). Being a right adjoint, the functor $\pt$ preserves
limits and hence lifts to internal categories. Hence we get a \emph{topological
behaviour category} $\bb{B}T := \nm{pt}(\nm{LB}T)$, which is just an
appropriately topologized
version of the behaviour category.

\begin{definition} (Topological behaviour category)
Let $T$ be a monad. The \emph{operational topology} on
$\bb{B}_1T$ is generated by subbasic opens of the form $[m|b] :=
\set{[m]_\beta | \beta \in [b]}$. Taking $\bb{B}_0T$ and $\bb{B}_1T$
with their operational topologies makes the structure maps of
the behaviour category continuous, yielding the \emph{topological
behaviour category} $\bb{B}T$.
\end{definition}

We have also shown that for finitary monads $T$, $\nm{LB}_0T$ is spatial and
hence corresponds to $\bb{B}_0T$. We ought to now establish also the spatiality
of $\nm{LB}_1T$, in order to conclude that it suffices to consider
the topological
behaviour category in the finitary monad case. Now, essentially $\nm{LB}_1T$ is
spatial because it is determined
by a sheaf $F_T$ over a spatial locale $\nm{LB}_0T$, and sheaves only
depend on the lattice of opens of its base space. More precisely, we
observe in the following lemma that the source map
$\sigma \colon \bb{B}_1T \to \bb{B}_0T$ of the topological
behaviour category is also a local homeomorphism,
and it has the same sections as the (source
map of the) localic behaviour category, and is thus the same sheaf.

\begin{lemma} \label{lemma:global-sections-of-B1T-is-F_T}
Let $T$ be a finitary monad. Then for any global section $s$ of
$\sigma \colon \bb{B}_1T \to \bb{B}_0T$, there is a finite family
of pairs $\set{(b_i \in \nm{BB}_0T, m_i \in T1)}_{i \in I}$ such
that $\set{b_i}_{i \in I}$ is a finite partition and $s$ maps
$\beta \in b_i \mapsto [m_i]_\beta$. Moreover, this family is
unique up to trace equivalence: if we have two such families
$\set{(b_i, m_i)}_{i \in I}$ and $\set{(b_j, n_j)}_{j \in J}$
then for all $i, j$ we have $m_i \sim_{b_i \wedge b_j} n_j$.
\end{lemma}

It is not hard to see that a family as in the lemma above defines
an element of the free $\nm{BB}_0T$-set, and that the uniqueness
translates to the same condition as in lemma
\ref{lemma:pointwise-trace-equivalence}. That is, the global
sections of $\bb{B}T$ correspond to the sheaf of transitions over
$\nm{BB}_0T$, and hence $\cal{O}(\bb{B}_1T) \cong \nm{LB}_1T$. Hence we have:

\begin{proposition} \label{prop:finitary-implies-spatiality}
For finitary $T$, $\nm{LB}_0T$ and $\nm{LB}_1T$ are
spatial, i.e., $\nm{LB}T$ has enough states and transitions.
\end{proposition}

\section{The Stone Adjunction for Monads}
\label{sec:stone-adjunction-for-monads}

In this section, we functorialize the construction of the localic
behaviour category and prove that it has a right adjoint by taking
sections of the source map. Here, the correct morphisms between
localic categories are \emph{retrofunctors}, not functors like usual.
Just as we view topological/localic categories as transition
systems, we can view a retrofunctor $\bb{C} \rightsquigarrow \bb{D}$
as a \emph{simulation} of transition systems.

\begin{definition} \textup{\cite[Example 2.9]{Clarke2020-tm}}
Let $\bb{C}$ and $\bb{D}$ be small categories. A \emph{retrofunctor} $F
\colon \bb{C} \rightsquigarrow \bb{D}$ consists of two functions
$F_0 \colon \bb{C}_0 \to \bb{D}_0$ and $F_1 \colon \bb{C}_0
\times_{\bb{D}_0} \bb{D}_1 \to \bb{C}_1$, where $\bb{C}_0
\times_{\bb{D}_0} \bb{D}_1$ is the pullback of $F_0$ along
$\sigma_{\bb{D}}$. In other words, given $c \in \bb{C}_0$ and $f
\in \bb{D}(F_0c, d)$ we get a \emph{lift} $F_1(c, f) \colon c \to
c'$ such that $F_0c' = d$. These are further required to respect identity
and composition:
\begin{equation}
  F_1(c, \id_{F_0c}) = \id_{c} \qquad F_1(c, g \circ f) = F_1(c',
  g) \circ F_1(c, f) \quad \text{ where } \quad F_1(c, f) \colon c \to c'
\end{equation}
If $\bb{C}$ and $\bb{D}$ are internal categories, then there is
an appropriate notion of internal retrofunctor which make the
appropriate diagrams commute, as described by Clarke \cite[definition
2.10]{Clarke2020-tm}. Write $\LocRetro$ and $\TopRetro$ for the
categories of internal categories and retrofunctors in $\Loc$ and
$\Top$ respectively.
\end{definition}

\begin{proposition} \label{prop:functoriality-of-LB}
The assignment $T \mapsto \nm{LB}T$ extends contravariantly to a
functor $\nm{LB} \colon \Mnd_r(\Set) \to \LocRetro^\op$, and
similarly a functor $\bb{B} \colon \Mnd_\omega(\Set) \to \TopRetro^\op$.
\end{proposition}
\begin{proof}
A monad morphism $\varphi \colon T \to S$ induces a retrofunctor
whose action on
objects $\nm{LB}_0\varphi \colon \nm{LB}_0S \to \nm{LB}_0T$ is given on
generating opens by $(\nm{LB}_0\varphi)\inv \colon [b] \mapsto
[\varphi(b)]$.
For the action on morphisms $(\nm{LB}_1\varphi) \colon \nm{LB}_0S
\times_{\nm{LB}_0T} \nm{LB}_1T \to \nm{LB}_1S$, by lemma
\ref{lemma:pullback-with-LB1T} we can identify
$\cal{O}(\nm{LB}_0S \times_{\nm{LB}_0T} \nm{LB}_1T)$ with an
appropriate poset of functions $h \colon T1 \to
\cal{O}(\nm{LB}_0S)$. The action is then simply given by
$(\nm{LB}_1\varphi)\inv
\colon w \mapsto w \circ \varphi_1$. We refer to
\cref{subsec:proof-of-functoriality-of-LB} for more details and
the verification of functoriality.
\end{proof}

On the other hand, if we view a localic category $\nm{LC}$ as a
transition system, what is a computation on $\nm{LC}$? Well, a
computation (of output type $A$) should specify, for each state $c
\in \nm{LC}_0$, a transition out of $c$ (the ``side-effect'' of the
computation) along with an output in $A$. In other words, the
computations are \emph{global sections} of the source map, or more
precisely a disjoint, $A$-indexed, jointly global family of partial
sections. Indeed, we get a monad of such sections---notice that this
looks very much like a state monad except we also keep track which
transitions are taken, not just the end state.

\begin{proposition} \label{prop:monad-of-sections}
Let $\nm{LC}$ be a localic category. Then the endofunctor
\[ \Gamma{\nm{LC}}(A) = \set{ s \colon \nm{LC}_0 \to A \cdot
    \nm{LC}_1 | \id_{\nm{LC}_0} = \nm{LC_0} \overset{s}{\to} A
    \cdot \nm{LC}_1 \overset{\pi_{\nm{LC}_1}}\to \nm{LC}_1
\overset{\sigma}{\to} \nm{LC}_0} \]
on $\Set$ (with the action on $f \colon A \to B$ given by
post-composing $f \cdot \nm{LC}_1$) admits a monad structure
given, for arbitrary $a \in A$, $s \in \Gamma\nm{LC}A, u \colon A
\to \Gamma\nm{LC}B$ by
$
\begin{tikzcd}[ampersand replacement=\&,cramped]
  {\return a = \nm{LC}_0} \& {\nm{LC}_1} \& {A \cdot \nm{LC}_1}
  \arrow["\iota", from=1-1, to=1-2]
  \arrow["{\upsilon_a}", from=1-2, to=1-3]
\end{tikzcd}$ and
\[
  \begin{tikzcd}[ampersand replacement=\&,cramped]
    {s \bind u = \nm{LC}_0} \& {A \cdot \nm{LC}_1} \&\&
    {\nm{LC}_1 \times (B \cdot \nm{LC}_1)} \& {B \cdot (\nm{LC}_1
    \times \nm{LC}_1)} \& {B \cdot \nm{LC}_1}
    \arrow["s", from=1-1, to=1-2]
    \arrow["{\braket{\pi, u \circ (A \cdot \tau)}}", from=1-2, to=1-4]
    \arrow["\cong", from=1-4, to=1-5]
    \arrow["\mu", from=1-5, to=1-6]
\end{tikzcd}.\]
In terms of points, we have $\return a \colon c \mapsto (a,
\id_c)$ and $s \bind u \colon c \mapsto (b, g \circ f)$ where
$(a, f) =: s(c)$ and $(b, g) =: u(a)(\tau(f))$, while in terms of
opens we have
\begin{alignat*}{3}
  &(\return a)\inv &&: \braket{a' \mapsto w} &&\mapsto \textup{if
  } a = a' \textup{ then } \iota\inv w \textup{ else } \bot \\
  &(s \bind u)\inv &&: \braket{b \mapsto w} &&\mapsto \bigvee_{a
  \in A} s\inv \Braket{a \mapsto \bigvee \set{v_1 \wedge \tau\inv
  u(a)\inv\braket{b \mapsto v_2} | v_1 \times v_2 \leq \mu\inv(w)} }.
\end{alignat*}
Moreover, the
assignment $\nm{LC} \mapsto \Gamma{\nm{LC}}$ defines a
contravariant functor $\Gamma \colon \LocRetro^\op \to \Mnd_r(\Set)$.
\end{proposition}
\begin{proofsketch}
A straightforward diagram chase reveals that unitality and
associativity of the monad structure is inherited from the
unitality and associativity of $\nm{LC}$.
It is a basic theorem of category theory that the functorial action
of a right adjoint functor is determined by the adjunction, so functoriality
is automatically obtained when we prove the
adjunction of theorem \ref{thm:stone-adjunction-loc}. Since we
don't use the functorial action of $\Gamma$, we leave
it as an exercise to define the action.
\end{proofsketch}

Any computation $t \in TA$ defines a global section $\eta(t) \colon
\nm{LB}_0T \to A \cdot \nm{LB}_1T$ of the behaviour category,
defined by $\eta(t)\inv \colon \braket{a \mapsto w} \mapsto [t
\mapsto a] \wedge w(t \rskip \return)$ on generating opens. This
defines the unit map of an adjunction between $\nm{LB}$ and
$\Gamma$, which brings us to the main adjunction of this paper.

\begin{theorem} \label{thm:stone-adjunction-loc}
$
\begin{tikzcd}[ampersand replacement=\&,cramped]
  {\nm{LB} : \Mnd_r(\Set)} \&\& {\LocRetro^\op : \Gamma}
  \arrow[""{name=0, anchor=center, inner sep=0}, shift left=2,
  from=1-1, to=1-3]
  \arrow[""{name=1, anchor=center, inner sep=0}, shift left=2,
  from=1-3, to=1-1]
  \arrow["\dashv"{anchor=center, rotate=-90}, draw=none, from=0, to=1]
\end{tikzcd}$.
\end{theorem}
\begin{proofsketch}
The counit map $\varepsilon \colon \nm{LC} \rightsquigarrow
\nm{LB}\Gamma\nm{LC}$ is given by $\varepsilon_0\inv \colon [s]
\mapsto s\inv\braket{1 \mapsto \top}$, and $\varepsilon_1\inv
\colon u \mapsto \lambda m \in T1. m\inv u$. See
\cref{subsec:proof-of-stone-adjunction-loc} for the verification
of the adjunction.
\end{proofsketch}

We also have a functor $\bb{B} := \pt \nm{LB} \colon \Mnd_r(\Set)
\to \TopRetro^\op$, and this similarly admits a right adjoint
$\Gamma$, but $\bb{B}T$ loses too
much information about the infinitary monad $T$, as exemplified by
example \ref{ex:locale-of-injective-state}. However, by proposition
\ref{prop:finitary-implies-spatiality}, if we restrict to $T \in
\Mnd_\omega(\Set)$ then $\nm{LB}T$ is spatial and corresponds to
the topological category $\bb{B}T$. The right adjoint of
the functor $\bb{B} \colon \Mnd_\omega(\Set) \to
\TopRetro^\op$ is given by taking the monad of \emph{finitary
sections} $\Gamma_\omega \bb{C}$ for a topological category $\bb{C}$.

\begin{theorem} \label{thm:stone-adjunction-top}
\begin{tikzcd}[ampersand replacement=\&,cramped]
  {\bb{B} : \Mnd_\omega(\Set)} \&\& {\TopRetro^\op : \Gamma_\omega.}
  \arrow[""{name=0, anchor=center, inner sep=0}, shift left=2,
  from=1-1, to=1-3]
  \arrow[""{name=1, anchor=center, inner sep=0}, shift left=2,
  from=1-3, to=1-1]
  \arrow["\dashv"{anchor=center, rotate=-90}, draw=none, from=0, to=1]
\end{tikzcd}
\end{theorem}
\begin{proof}
The inclusion $i \colon \Mnd_\omega(\Set) \hookrightarrow
\Mnd_r(\Set)$ has a right adjoint given by $\lan_j(- \circ j)$
where $j \colon \Set_\omega  \hookrightarrow \Set$ is the usual
full inclusion of the category of finite sets (as follows from
relative monad theory \cite{Altenkirch2010-ub}). Then we have the
desired adjunction by composing $\bb{B} \dashv \Gamma$ with $i
\dashv \lan_j(- \circ j)$, noting that $\Gamma_\omega :=
\lan_j(\Gamma(-) \circ j)$.
\end{proof}

\section{The Stone Duality for Hyperaffine-Unary Monads}
\label{sec:stone-duality-for-haffun-monads}

This final section is devoted to proving that adjunctions
\ref{thm:stone-adjunction-loc} and \ref{thm:stone-adjunction-top} are
idempotent, and to characterize
their fixed points. On the monad side, the fixed points correspond to
those monads with a Cartesian-closed category of Eilenberg-Moore algebras.
These were first syntactically characterized by
Johnstone \cite{Johnstone1990-dc}, but the first-named author
provided an improved characterization \cite{Garner2024-yc}
as those monads
which admit a \emph{hyperaffine-unary decomposition}. On the side
of localic categories, the fixed points are the \emph{ample localic
categories}, i.e., whose source maps are local homeomorphisms and
whose locales of objects are ultraparacompact.

In fact, this equivalence between hyperaffine-unary monads and
ample localic categories is originally due to the first-named
author \cite{Garner2025-hr}. In addition to filling in details
about the equivalence,
our contribution is to envelope the equivalence in an adjunction,
which yields a
process of \emph{hyperaffine-unary completion} for monads on one
hand, and a process of \emph{amplification} for localic categories
on the other.

\subsection{Hyperaffine-Unary Monads}

What exactly is the difference between $T$ and $\Gamma\nm{LB}T$?
Because of the unit map, all computations in $T$ live inside
$\Gamma\nm{LB}T$. The answer is that $T$ adds additional operations
$\overline{t}$ which predict the output of $t$ without performing
the side effect of $t$. Computationally, this can be thought of as
performing $t$ and then rolling the state back, which we will refer
to as \emph{scrying}. In the case where $T$ encodes computations
consuming from a stream\footnote{A particularly potent analogy is to think of
the environment as an (infinite) deck of playing cards, and of the
program as the player, in which case a scry allows the player to
look at the top $n \in \bb{N}$ cards of their deck before putting
it back in the same order. Players of a certain popular trading card game
will recognize this, but note that ``scrying'' there allows you to
reorder the cards (and to put them at the bottom of the deck).}, the
newly added scrying operations look ahead into the stream without
consuming from it, as the following example demonstrates.

\begin{example}(Binary Input)
Let $T$ be the monad of binary input, whose terminal topological comodel
$\bb{B}_0T$ is the cantor space of example \ref{ex:cantor-space}. By
continuity and
compactness of $\bb{B}_0T$, any global section $s \in
\Gamma\bb{B}T(A)$ is therefore described (non-uniquely) by a pair
$(B, |s| \colon B \to \bb{N} \times A)$ where $B$ is a finite set
of finite strings $B \subseteq 2^{< \omega}$ that jointly cover
all infinite streams, and $|s|$ assigns to each element of $B$ a
pair $(n, a)$ consisting of the number $n$ of digits to consume
from the stream, and the output $a$. As an example, the binary tree
$t = b(a_0, b(b(a_1, a_2), a_3)) \in TA$, induces a section
$\eta(t)$ which is described by the assignments $\set{0 \mapsto
(1, a_0); 100 \mapsto (3,a_1); 101 \mapsto (3, a_2); 11 \mapsto (2, a_3)}$.

In general, an assignment $\frk{b} \mapsto (n, a)$ for a section
of the form $\eta(t)$ must satisfy $n \geq \nm{length}(\frk{b})$.
That is, to use information about the first $k =
\nm{length}(\frk{b})$ digits of the stream you must consume at
least $k$ digits. However, in general sections do not need to
respect this: the assignments $\set{0 \mapsto (0, a_0); 10
\mapsto (1, a_1); 11 \mapsto (1, a_2)}$ describe a perfectly
acceptable section $s$. We can think of $s$ as \emph{looking
ahead at} or \emph{scrying} the first two digits of the stream,
before deciding what to do. Indeed, for any section $s$, we have
a corresponding $\overline{s} \in \Gamma\bb{B}T(A)$ which outputs
the same values as $s$, but only makes identity transitions. Then
$s$ factors as $s = \overline{s} \bind \lambda a. s \rskip \return a$.
\end{example}

It is easy to see in general that monads of the form
$\Gamma\nm{LC}$ always have this factorization property, since a
section $s$ always admits a cousin $\overline{s}$ which sends
objects to the same output, but replaces the morphism by identity
morphisms---the \emph{scry} corresponding to $s$.
Monads satisfying this factorization property are called
\emph{hyperaffine-unary theories} in \cite{Garner2024-yc}, and they suffice
to characterize the fixed points of adjunction \ref{thm:stone-adjunction-loc}.

\begin{definition} (Hyperaffine-unary)
Let $T$ be a monad. A computation $h \in TA$ is \emph{hyperaffine} if
\begin{equation}
  h \rskip \return a = \return a \qquad\text{ and }\qquad h \bind
  \lambda a_1. h \bind \lambda a_2. \return (a_1, a_2) = h \bind
  \lambda a. \return (a, a). \label{eqn:hyperaffine-defn}
\end{equation}
The monad $T$ is \emph{hyperaffine-unary} if for every
computation $t \in TA$, there is a unique hyperaffine
$\overline{t} \in TA$ such that $t = \overline{t} \bind \lambda
a. (t \rskip \return a)$. Any hyperaffine-unary monad admits a
submonad $H$ of hyperaffine operations (\cite[Proposition
6.1]{Garner2024-yc}).
\end{definition}

\begin{proposition} \label{prop:global-sections-monad-is-haffun}
Let $\nm{LC}$ be a localic category. Then $\Gamma\nm{LC}$ is
hyperaffine-unary.
\end{proposition}

Hyperaffine-unary monads admit a particularly nice presentation of
the localic behaviour category, which greatly aids us in proving
the characterization of the fixed points.

\begin{lemma} \label{lemma:H2-of-hyperaffine-unary-monad}
Let $T$ be a hyperaffine-unary monad with submonad of hyperaffines
$H$. Then $\cal{O}(\nm{LB}_0T)$ is generated by $H2_\cal{J}$
where $H2$ admits a Boolean algebra structure and $\cal{J}$ is a
strongly zero-dimensional topology defined by $\set{P^{(h)} | h
\in HA}$ with $P^{(h)} = \set{[h \mapsto a] | a \in A}$. Here, we
abuse notation by writing $[h \mapsto a]$ for $h \bind \lambda
a'. \delta_a(a') \in H2$. Moreover, the map $\delta \colon T1 \to
F_T$ is an isomorphism, with $T1$ admitting a $H2_{\cal{J}}$-set
structure whose $P^{(h)}$-ary operations are defined by
$P^{(h)}(\lambda a. m_a) := h \bind \lambda a. m_a$.
\end{lemma}

\begin{proposition} \label{prop:left-fixpoints-are-haffun}
A monad $T$ is hyperaffine-unary iff the unit map $\eta_T
\colon T \to \Gamma\nm{LB}T$ of adjunction
\ref{thm:stone-adjunction-loc} is an isomorphism.
\end{proposition}

\subsection{Ample Localic Categories and Stone Duality}

On the other side of the adjunction, what is the relationship
between a localic category $\nm{LC}$ and the behaviour category
$\nm{LB}\Gamma\nm{LC}$? Well, $\Gamma$ only considers the
partitioning sections of $\nm{LC}$, so is only sensitive to the
\emph{ultraparacompact quotient} of $\nm{LC}_0$, i.e., whose frame
of opens is the ultraparacompact frame generated by taking the
zero-dimensional topology of partitions on the Boolean algebra
$\frk{B}(\nm{LC}_0)$. Moreover, $\nm{LB}$ then reconstructs the
locale of morphisms from only the sections over this
ultraparacompact quotient. So this reconstruction is perfect if in
the first place $\nm{LC}_0$ is ultraparacompact and the source map
is a local homeomorphism. These are called \emph{ample localic
categories} \cite{Garner2025-hr}.

\begin{definition} \label{defn:ample-category}
A localic category $\nm{LC}$ is \emph{ample} if
$\sigma_{\nm{LC}}$ is a local homeomorphism and $\nm{LC}_0$ is
ultraparacompact. A topological category $\bb{C}$ is \emph{ample}
if $\sigma_{\bb{C}}$ is a local homeomorphism and $\bb{C}_0$ is a
Stone space. Write $\AmpLocRetro$ (resp. $\AmpTopRetro$) for the full
subcategory of $\LocRetro$ (resp. $\TopRetro$) containing the ample
localic (resp. topological) categories.
\end{definition}

\begin{proposition} \label{prop:right-fixpoints-are-ample}
A localic category $\nm{LC}$ is ample iff the counit map
$\varepsilon_{\nm{LC}} \colon \nm{LC} \rightsquigarrow
\nm{LB}\Gamma\nm{LC}$ is an isomorphism.
\end{proposition}

The combination of propositions
\ref{prop:left-fixpoints-are-haffun} and
\ref{prop:right-fixpoints-are-ample} allow us to derive the titular
Stone duality for monads.

\begin{theorem} \label{thm:stone-duality}
The adjunction of theorem \ref{thm:stone-adjunction-loc} is
idempotent and its fixed points are the equivalent categories
$\nm{HUMnd_r} \simeq \AmpLocRetro$. Furthermore, this equivalence
restricts to $\nm{HUMnd_\omega \simeq \AmpTopRetro}$.
\end{theorem}

\begin{example}
Every Grothendieck Boolean algebra $B_{\cal{J}}$ presenting a
locale $L$ is associated to a distributions monad $T_{B_{\cal{J}}}(A) :=
A[B]^\cal{J}$ (\ref{remark:equational-presentation-of-TB}) whose
computations are all hyperaffine, and hence
hyperaffine-unary. Under the equivalence of theorem
\ref{thm:stone-duality}, $T_{B_{\cal{J}}}$ corresponds to the localic
category $\nm{LB}T_{B_{\cal{J}}}$ with $\nm{LB}_0T_{B_{\cal{J}}}
\cong \nm{LB}_1T_{B_{\cal{J}}} \cong L$ and source map $\sigma = \id
\colon L \to L$. Of course, if $\cal{J}$
consists of only finite partitions, then $L = \nm{Spec}(B)$ is
the Stone dual of $B$, and so our equivalence subsumes the
classical Stone duality.
\end{example}

\newpage
\bibliographystyle{./entics}
\bibliography{refs}

\newpage

\appendix

\section{Omitted Proofs}

\subsection{Proof of proposition
\ref{prop:terminal-comodel-of-partitions-monad}}

This admits a rather lengthy ``elementary'' proof, but we can also prove it
from other results in this paper, namely via the following chain of
isomorphisms:
\begin{equation*}
\bb{B}_0T_B \overset{\text{prop
\ref{prop:loc-top-adjunction-for-comodels}}}{\cong} \pt \nm{LB}_0T
=     \Loc(1, \nm{LB}_0T_B)
\overset{\text{lem.
\ref{lemma:H2-of-hyperaffine-unary-monad}}}{\cong} \Frm(\nm{Idl}(B), S)
\cong \nm{BA}(B, S)
\cong \nm{\nm{Spec}(B)}
\end{equation*}
where $S = \cal{O} 1$ is the initial frame $\set{\bot \leq \top}$,
the penultimate
isomorphism follows by universal property of the ideal construction
as the left adjoint
to the inclusion $BA \to \Frm$, and the last
isomorphism identifies a map $p : B \to S$ with the ultrafilter
$\frk{p} = \set{b | p(b) = \top}$.

\subsection{Proof of proposition \ref{prop:terminal-localic-comodel}}
We first check that $\codeno{-}$ respects $\bind$ and $\return$,
making $\nm{LB}_0T$ a comodel. For $\bind$ we have
\begin{alignat*}{2}
& \codeno{t \bind u}\inv\braket{b_0 \mapsto [t_0]} \\
= & [t \bind u \mapsto b_0] \wedge [t \bind u \rskip t_0] &&
\text{by definition of } \codeno{t \bind u} \\
= & \bigvee_{a \in A} [t \mapsto a] \wedge [t \rskip u(a) \mapsto
b_0] \wedge [t \rskip u(a) \rskip t_0] \qquad && \text{apply
\eqref{eqn:LB-bind} twice, and simplify} \\
= & \bigvee_a ([t \mapsto a] \wedge [t \rskip u(a) \mapsto b_0])
\wedge ([t \mapsto a] \wedge [t \rskip u(a) \rskip t_0]) \qquad && \\
= & \bigvee_{a} \codeno{t}\inv\braket{a \mapsto [u(a) \mapsto
b_0]} \wedge \codeno{t}\inv \braket{a \mapsto [u(a) \rskip t_0]}
&& \text{by definition of } \codeno{t} \\
= & \codeno{t}\inv \bigvee_{a} \braket{a \mapsto [u(a) \mapsto
b_0] \wedge [u(a) \rskip t_0]} && \text{by definition of } \codeno{u} \\
= & \codeno{t}\inv\codeno{u}\inv \braket{b_0 \mapsto [t_0]},
\end{alignat*}
whereas for $\return$, we compute
\[ \codeno{\return a} \braket{a_0 \mapsto [t_0]} = [\return a
\mapsto a_0] \wedge [\return a \rskip t_0] =
\begin{cases}
  [t_0] & \text{if } a = a_0 \\
  \bot & \text{otherwise,}
\end{cases} \]
but this is just $\upsilon_{a}$. Next, we show that this comodel is
terminal, so let $L$ be an arbitrary localic comodel. If a map $h
\colon L \to \nm{LB}_0T$ exists, then $h$ being a comodel map
implies $h\inv [t_0] = h\inv\codeno{t_0}\inv\braket{1 \mapsto \top}
= \codeno{t_0}_L\inv (2 \cdot h)\inv \braket{1 \mapsto \top} =
\codeno{t_0}_L\inv \braket{1 \mapsto \top}$, and hence this
uniquely determines $h$. We leave the verification that this map is
well-defined as an exercise to the reader.

\subsection{Proof of correspondence for example
\ref{ex:locale-of-injective-state}} \label{subsec:proof-of-injective-state}

\begin{definition} \label{defn:locale-of-injective-functions}
The locale of injective functions $\bb{R} \rightarrowtail \bb{N}$
is presented by generators $\braket{x \mapsto n}$ for $x \in
\bb{R}$ and $n \in \bb{N}$, required to satisfy, for $x \neq y$
and $m \neq n$, the equations
\[ \bigvee_{x \in \bb{R}} \braket{x \mapsto n} = \top \qquad
  \braket{x \mapsto n} \wedge \braket{x \mapsto m} = \bot \qquad
\braket{x \mapsto n} \wedge \braket{y \mapsto n} = \bot \]
\end{definition}

We must prove that this presentation is bi-interpretable with the
behaviour locale of injective state. In one direction, the
generator $\braket{x \mapsto n}$ is interpreted as $[\nm{get}_x
\mapsto n]$, in which case the first two axioms straightforwardly
follow. The third axiom can be proven as follows:
\begin{alignat*}{2}
& [\nm{get}_x \mapsto n] \wedge [\nm{get}_y \mapsto n] \\
=~& [\nm{get}_x \mapsto n] \wedge [\nm{get}_x \rskip \nm{get}_y
\mapsto n] \\
=~& [\nm{get}_x \bind \lambda m_1. \nm{get}_y \bind \lambda m_2.
\return (m_1, m_2) \mapsto (n, n)] \\
=~& [\nm{get}_x \bind \lambda m_1. \nm{get}_y \bind \lambda m_2.
\return (m_1 \overset{?}{=} n \overset{?}{=} m_2)] \\
=~& [\nm{get}_x \bind \lambda m_1. \nm{get}_y \bind \lambda m_2.
\return 0] && ~~~\text{(by injectivity eqn)} \\
=~& [\return 0] = \bot
\end{alignat*}
In the other direction, by recursion we define a map $h$
interpreting the generators of the behaviour locale:
\[ h \colon [\return 0] \mapsto \bot \qquad [\return 1] \mapsto
\top \qquad [\nm{get}_x \bind \lambda n. t_n] \mapsto  \bigvee_{n
\in \bb{N}}
\braket{x \mapsto n} \wedge h([t_n]) \]
We leave the reader to verify that this respects the usual equations satisfied
by terms in the theory of state. As an example, we will verify just
the injectivity equation from example \ref{ex:locale-of-injective-state}:
\begin{alignat*}{2}
& h[\nm{get}_x \bind \lambda n. \nm{get}_y \bind \lambda m.
\return f(n, m) \mapsto (n_0, m_0)] \\
= & \bigvee_{n} \bigvee_{m} \braket{x \mapsto n} \wedge \braket{y
\mapsto m} \wedge f(n, m) \overset{?}{=} (n_0, m_0) \\
= & \left(\bigvee_{n \neq m} \braket{x \mapsto n} \wedge
  \braket{y \mapsto m} \wedge (n, m) \overset{?}{=} (n_0, m_0)
\right)\vee \bigvee_{n} \braket{x \mapsto n} \wedge \braket{y
\mapsto n} \wedge f(n, m) \overset{?}{=} (n_0, m_0) && \text{
property of } f \\
= & \left(\bigvee_{n \neq m} \braket{x \mapsto n} \wedge
  \braket{y \mapsto m} \wedge (n, m) \overset{?}{=} (n_0, m_0)
\right) \vee \bot && \text{ by def.
\ref{defn:locale-of-injective-functions}} \\
= & \left(\bigvee_{n \neq m} \braket{x \mapsto n} \wedge
  \braket{y \mapsto m} \wedge (n, m) \overset{?}{=} (n_0, m_0)
\right)\vee \bigvee_{n} \bot \wedge (n, m) \overset{?}{=} (n_0, m_0) \\
= & \left(\bigvee_{n \neq m} \braket{x \mapsto n} \wedge
  \braket{y \mapsto m} \wedge (n, m) \overset{?}{=} (n_0, m_0)
\right)\vee \bigvee_{n} \braket{x \mapsto n} \wedge \braket{y
\mapsto n} \wedge (n, m) \overset{?}{=} (n_0, m_0) \\
= & \bigvee_{n} \bigvee_{m} \braket{x \mapsto n} \wedge \braket{y
\mapsto m} \wedge (n, m) \overset{?}{=} (n_0, m_0) \\
= & h[\nm{get}_x \bind \lambda n. \nm{get}_y \bind \lambda m.
\return (n, m) \mapsto (n_0, m_0)]
\end{alignat*}

Finally, $h$ respects the equations of the behaviour
locale: for all three axioms \eqref{eqn:LB-bot},
\eqref{eqn:LB-return}, and \eqref{eqn:LB-bind} it follows by a
straightforward induction on the syntax of $t$.

\subsection{Proof of proposition \ref{prop:LB0T-ultraparacompact}}
\label{subsec:proof-of-LB0T-ultraparacompact}

We follow a similar line of argument to \cite[Section
2]{Johnstone1997-ry}. First, notice that we can construct
$\nm{LB}_0T$ in two steps. Begin by constructing the meet
semi-lattice $\nm{MB_0T}$ generated by opens $[b]$ subject to
equations $[t \rskip \return a \mapsto a] = \top$ and $[t \mapsto
a] \wedge [t \bind u \mapsto b] = [t \mapsto a] \wedge [t \rskip
u(a) \mapsto b]$. Then we can generate the frame
$\cal{O}\nm{LB}_0T$ from $\nm{MB}_0T$ subject to the equations $[t
\mapsto a] \wedge [t \mapsto a'] = \bot$ and $\bigvee_{a \in A} [t
\mapsto a] = \top$. Following \cite[II 2.11]{Johnstone1982-ei},
this can be presented as a covering system instead. The covering
system $\cal{J}$ is generated by the following rules:
\begin{mathpar}
\inferrule{t \in TA, a \neq a' \in A}{\emptyset \in \cal{J}([t
\mapsto a] \wedge [t \mapsto a'])} \quad
\inferrule{t \in TA}{\set{[t\mapsto a] | a \in A} \in
\cal{J}(\top)} \quad \inferrule{ }{\set{u} \in \cal{J}(u)} \\
\inferrule{J \in \cal{J}(u)}{\set{j \wedge v | j \in J} \in
\cal{J}(u \wedge v)} \quad \inferrule{J \in \cal{J}(u) \quad K_j
\in \cal{J}(j) \;\; \forall j \in J}{\bigcup_{j \in J} K_j \in \cal{J}(u)}
\end{mathpar}
Notice that the three base cases are pairwise-disjoint covers, and
the inductive cases preserve the pairwise-disjoint property. Hence
all the covers in this system are pairwise-disjoint. The frame
presented by this system consists of all the $\cal{J}$-ideals, i.e.,
downwards closed subsets $I \subseteq \nm{MB}_0T$ such that for any
$J \in \cal{J}(u)$, $J \subseteq I$ implies $u \in I$. For any
subset $S \subseteq \nm{MB}_0T$, the smallest $\cal{J}$-ideal
containing $S$ is $\overline{S} = \set{u \in \nm{MB}_0T | \exists J
\in \cal{J}(u). J\subseteq S}$, and the join of a family of
$\cal{J}$-ideals $\set{I_{k}}_k$ is the $\overline{\bigcup_k I_k}$.
Also, every $u \in \nm{MB}_0T$ induces an ideal $\downarrow u$.

We now prove that every open cover is refined by a partition, so
consider an open cover $\set{I_k}_k$. It is covering iff $\top \in
\overline{\bigcup_k I_k}$ iff there is $J \in \cal{J}(\top)$ such
that $J \subseteq \bigcup_k I_k$. But if we now consider the family
$\set{\downarrow j | j \in J}$, then this is precisely a partition
(pairwise-disjoint because $J$ is) refining $\set{I_k}_k$.

\subsection{Proof of proposition
\ref{prop:compact-and-ultraparacompact-is-stone}}

Suppose $L$ is compact and ultraparacompact.
This ensures that $L$ is freely generated
by the Boolean algebra $\frak{B}(L)$, and locales
freely generated from a distributive lattice in this way are
well-known to be spatial \cite[II 3.4]{Johnstone1982-ei}. $\pt L$
inherits zero-dimensionality
and compactness from $L$, and it is Hausdorff because for any two
points $\frak{p} \neq \frak{q} \in \pt L$, there must be a $b \in
\frak{B}(L)$ such that $\frak{p} \in [b]$ and $\frak{q} \in [\neg
b]$, thus separating the two points.

On the other hand, if $L =
\cal{O}X$ for a Stone space $X$, then the compactness ensures that
every open cover has a finite subcover, but any finite
cover is further refinable into a finite partition by
zero-dimensionality. This establishes
the desired ultraparacompactness.

\subsection{Proof of lemma \ref{lemma:basis-of-etale-space}}
\label{subsec:proof-of-basis-of-etale-space}

Each open $\hat{x} = \lambda y. \bigvee \set{b | x \equiv_b y}$ is
a $B_\cal{J}$-set homomorphism since if $y_1 \equiv_c y_2$ then
\[c \wedge \hat{x}(y_1) = \bigvee\set{c \wedge b | x \equiv_b y_1}
\overset{\Diamond}{=} \bigvee \set{c \wedge b | x \equiv_{c \wedge
b} y_1} = \bigvee \set{c \wedge b | x \equiv_{c \wedge b} y_2} = c
\wedge \hat{x}(y_2), \]
where $\Diamond$ holds from right-to-left because we can take the
$b$ on the LHS to be the $b \wedge c$ on the RHS. The assignment $x
\mapsto \hat{x}$ is injective, as we now show. Let $x, y \in |F|$
and assume that $\hat{x} = \hat{y}$. Then we have $\bigvee_b \set{b
| x \equiv_b y} = \hat{x}(y) = \hat{y}(y) = \top$, which implies
there is a (cover $P$, WLOG refinable into a) partition $P$ such
that for each $b \in P$ we have $x \equiv_b y$. But then $y =
P(\lambda b. y) = P(\lambda b. b(x, y)) = P(\lambda b. x) = x$.

Finally, let us show that $w = \bigvee_{x \in |F|} \hat{x} \wedge
\nm{const}_{w(x)}$. Take an arbitrary $y \in |F|$. Then $w(y) =
\top \wedge w(y) = \hat{y}(y) \wedge w(y) \leq \bigvee_{x \in |F|}
\hat{x}(y) \wedge \nm{const}_{w(x)}(y)$. On the other hand, for the
right-to-left inequality, it suffices to show for any $b \in B$
with $x \equiv_b y$ that $b \wedge w(x) \leq w(y)$. But this is
clearly true since $x \equiv_b y$ implies $w(x) \equiv_b w(y) \iff
b \wedge w(x) = b \wedge w(y)$.

\subsection{Proof of proposition
\ref{prop:corresponding-local-homeo}}
\label{subsec:proof-of-corresponding-local-homeo}

The map $\sigma \colon E(F) \to L$ is a local homeomorphism by
taking the family $\set{\hat{x}}_{x \in |F|}$, which is covering by
\ref{lemma:basis-of-etale-space}. Each such open $\hat{x}$ is
homeomorphic onto the whole base space $\cal{O}L$.

The set of global sections to $E(F) \to L$ forms a $B_\cal{J}$-set,
$s_1 \equiv_b s_2 \iff \forall w \in \cal{O}E(F). s_1(w) \wedge b =
s_2(w) \wedge b$. Every element $x \in |F|$ corresponds to a global
section $s_x \colon L \to E(F)$ given by $s_x\inv \colon w \mapsto
w(x)$. This defines a $B_\cal{J}$-set homomorphism since $x
\equiv_b y$ implies, for all $w \in \cal{O}E$ that $s_y(w) \wedge b
= w(y) \wedge b = w(b(x, y)) \wedge b = (b \wedge w(x) \vee \neg b
\wedge w(y)) \wedge b = b \wedge w(x)$, i.e., that $s_x \equiv_b
s_y$. Moreover, this assignment is \emph{injective} by the
injectivity of $x \mapsto \hat{x}$ (lemma \ref{lemma:basis-of-etale-space}).

On the other hand, to see that this assignment is
\emph{surjective}, notice that a global section $s$ given by a
frame homomorphism $s\inv \colon \cal{O}E(F) \to \cal{O}L$ induces
a map $C \colon x \mapsto s\inv(\hat{x}) \colon |F| \to \cal{O} L$.
The image of this map covers $\cal{O}L$ because when we take all
the opens $\hat{x}$ together, they cover $\cal{O}E(F)$. Since $L$
is ultraparacompact, we can refine $C$ to an (extended) partition
$P \colon |F| \to \cal{O}L$ (not uniquely, but we don't care which
one we choose). This yields an element $p \in |F|$ by taking the
amalgamation $p := P[|F|]^{-}(\lambda b. P\inv b)$. The element $p$
induces the section $s_p$, so to see that $s_p = s$:
\begin{alignat*}{1}
s\inv_p(w) = &\; w(p) = w(P[|F|]^{-}(\lambda b. P\inv b)) \\
= & \bigvee_{b \in P} b \wedge w(P\inv b) = \bigvee_{x \in |F|}
P(x) \wedge w(x) \\
= & \bigvee_{x \in |F|} C(x) \wedge w(x) = \bigvee_{x \in |F|}
s\inv(\hat{x}) \wedge w(x) = s\inv(w)
\end{alignat*}
whereby the last equality follows from lemma \ref{lemma:basis-of-etale-space}.

\subsection{Proof of proposition \ref{prop:germs-of-local-homeo}}
\label{subsec:proof-of-germs-of-local-homeo}
Given a germ $[x]_p$, we can define a point $q \in \pt E(F) \cong
\Set_{B_\cal{J}}(F, L) \to \cal{O}1$ by letting $q(w) = \top$ iff
$p \in w(x)$. This is coherent with respect to the choice of $x$,
because if $x \equiv_p y$ then $x \equiv_b y$ for some $b \ni p$,
so $w(x) \equiv_b w(y)$ which entails $b \wedge w(x) = b \wedge
w(y)$ by definition. But then $p \in w(x) \iff p \in b \wedge w(x)
\iff p \in b \wedge w(y) \iff p \in w(y)$.

On the other hand, suppose we have a point $q \colon 1 \to E(F)$.
This defines a subset $\hat{q} \subseteq |F| = \set{x | q \in
\hat{x}}$, which has to be non-empty because otherwise $q \not\in
\hat{x}$ for all $x \in |F|$, and hence $q\inv(\top) = \bigvee_{x
\in |F|} q\inv(\hat{x}) = \bot$ contradicting $q\inv$ being a frame
homomorphism. $q$ also induces a point $p := \sigma q$, and the
germ induced by $q$ is $[x]_p$ for any $x \in \hat{q}$. This is
coherent: for any $x, y \in \hat{q}$ we have $q \in \hat{x} \wedge
\hat{y}$, but by lemma \ref{lemma:basis-of-etale-space} there is
some $z \in |F|$ such that $q \in \hat{z}$ and $p \in (\hat{x}
\wedge \hat{y})(z)$. Then
\begin{alignat*}{2}
p \in (\hat{x} \wedge \hat{y})(z) &\iff  p \in \set{b \wedge b' |
x \equiv_b z, z \equiv_{b'} y} \\
& \implies \exists b \wedge b' \ni p. x \equiv_{b \wedge b'} z
\text{ and } z \equiv_{b \wedge b'} y \\
& \implies \exists b'' \ni p. x \equiv_b'' y \\
& \iff x \equiv_p y
\end{alignat*}
From here it is a routine unfolding of definitions to see that a
germ induces itself by going back-and-forth. On the other hand,
given a point $q$, going forth-and-back produces point $q'$ with
$q' \in w$ iff $q \in \nm{const}_{w(x)}$ for some $x \in \hat{q}$
iff $q \in \bigvee_{x} \hat{x} \wedge \nm{const}_{w(x)} = w$, and
hence $q' = q$.

The subbasic opens $[x | b]$ on the set of germs is induced by the
opens $\hat{x} \wedge \nm{const}_b$, which generate all other opens
by lemma \ref{lemma:basis-of-etale-space}.

\subsection{Proof of lemma \ref{lemma:pointwise-trace-equivalence}}
\label{subsec:proof-of-pointwise-trace-equivalence}

First, for self-containedness we lay out precisely the definition
of $B_\cal{J}$-congruence.
\begin{definition}
Let $X$ be a $B_{\cal{J}}$-set. An equivalence relation
$\mathop{\approx} \subseteq X \times X$ is a $B_{\cal{J}}$-set
congruence if for every partition $P \in \cal{J}$, and two
families $x, x' \colon P \to X$ such that for each $b$, $x_b
\approx x'_b$, we have
\[ P(x) ~R~ P(x'). \]
Given a set of pairs $G \subseteq X \times X$, The congruence
$\approx_G$ generated by $G$ is given by the following inference rules
\begin{mathpar}
  \inferrule[gen]{(x_1, x_2) \in G}{x_1 \approx_G x_2} \qquad
  \inferrule[refl]{x \in X}{x \approx_G x} \qquad
  \inferrule[trans]{x_1 \approx_G x_2 \\ x_2 \approx_G x_3}{x_1
  \approx_G x_3} \qquad \inferrule[symm]{x_1 \approx_G x_2}{x_2
  \approx_G x_1} \\
  \inferrule[cong-$P$]{x_b \approx_G x'_b \quad \forall b \in
  P}{P(x) \approx_G P(x')}
\end{mathpar}
We also define $\rightsquigarrow_G$ as the relation derivable by
exactly one use of \textsc{gen}, any use of \textsc{cong}-$P$ for
any \emph{finite} partition $P$, and any use of \textsc{refl}.
\end{definition}

If $\cal{J}$ only contains finite partitions, then the algebraic
theory of $B_{\cal{J}}$-sets only has finite operations and we can
easily show that $\approx_G$ is the symmetric transitive closure of
$\rightsquigarrow_G$. However, if $\cal{J}$ has infinite
partitions, then this is no longer the case, but we can still prove
that $\approx_G$ is always derivable by one congruence applied to
$\leftrightsquigarrow_G^\omega$, i.e., the transitive symmetric
closure of the relation $\rightsquigarrow_G$.

\begin{lemma} \label{lemma:partition-associated-to-rewrite}
If $x_1 \rightsquigarrow_G x_2$, then this can be derived with
exactly one application of the \textsc{cong} rule.
\end{lemma}
\begin{proof}
By induction on derivation of $x_1 \rightsquigarrow_G x_2$. In
the base case, we clearly have zero applications, but we can add
an application of the $\textsc{cong}$ rule over the one-element
partition $\set{\top}$. In the inductive case, we have a
derivation which looks like
\begin{mathpar}
  \inferrule{x_b \rightsquigarrow_G y_b \quad \forall b \in
  P}{P(x) \rightsquigarrow_G P(y)}.
\end{mathpar}
By the inductive hypothesis, each subderivation of $x_b
\rightsquigarrow_G y_b$ can be rewritten to use exactly one
\textsc{cong} rule, so each subderivation is associated with a
partition $Q_b$, defining a map $Q \colon P \to \cal{J}$. The
derivation now looks like the derivation on the left-hand side
below, which is derivable as on the right-hand side.
\begin{mathpar}
  \inferrule*[Left={cong}]{
    \inferrule*[Left={cong}]{
      \inferrule*[Left={gen or refl}]{
        \ldots
      }{x^c_b \rightsquigarrow_G y^c_b} \quad \forall c \in Q_b
    }{x_b = Q_b(\lambda c. x^c_b) \rightsquigarrow_G Q_b(\lambda
    c.y^c_b) = y_b} \quad \forall b \in P
  }{P(x) \rightsquigarrow_G P(y)}
  \qquad \qquad \quad
  \inferrule*[Left={cong}]{
    \inferrule*[Left={gen or refl}]{
      \ldots
    }{x^c_b \rightsquigarrow_G y^c_b} \quad \forall (b \wedge c) \in P;Q
  }{P;Q(\lambda b \wedge c. x^c_b) \rightsquigarrow_G P;Q(\lambda
  b \wedge c. y^c_b)}
\end{mathpar}
The right-hand derivation uses only one \textsc{cong}, concluding the proof.
\end{proof}

\begin{lemma} \label{lemma:normal-form-of-approxG}
If $x_1 \approx_G x_2$ then this is derivable by a derivation of the form
\begin{mathpar}
  \inferrule*[Left={cong}]{
    \inferrule*{\vdots}{x_b \leftrightsquigarrow^\omega_G y_b}
    \quad \forall b \in P
  }{x_1 = P(x) \approx P(y) = x_2}
\end{mathpar}
\end{lemma}
\begin{proof}
By induction on the derivation of $x_1 \approx_G x_2$. The base
cases and the inductive cases for \textsc{symm} and
\textsc{cong}-$P$ are easy, so we only work out the case for
\textsc{trans}. By induction hypothesis we know that our
derivation will look like
\begin{mathpar}
  \inferrule*[Left={trans}]{
    \inferrule*[Left={cong}]{
      \inferrule*{\vdots^{\Delta^x_b}}{x_b
      \leftrightsquigarrow^\omega_G x'_b} \quad \forall b \in P
    }{P(x) \approx_G P(x')}
    \qquad\qquad
    \inferrule*[Left={cong}]{
      \inferrule*{\vdots^{\Delta^y_c}}{y'_c
      \leftrightsquigarrow^\omega_G y_c} \quad \forall c \in Q
    }{Q(y') \approx_G Q(y)}
    \qquad\qquad
    P(x') = Q(y')
  }{x_1 = P(x) \approx_G Q(y) = x_2}
\end{mathpar}
We can re-arrange this into the following derivation,
\begin{mathpar}
  \inferrule*[Left={cong}]{
    \inferrule*[Left={trans}]{
      \inferrule*[Left={cong}]{
        \inferrule*{\vdots^{\Delta^x_b}}{x_b
        \leftrightsquigarrow^\omega_G x'_b}
        \quad
        \inferrule*[Right={refl}]{ }{* \approx_G *}
      }{(b \wedge c)(x_b, *) \approx_G (b \wedge c)(x'_b, *)}
      \qquad\quad
      \inferrule*[Left={cong}]{
        \inferrule*{\vdots^{\Delta^y_c}}{y_c'
        \leftrightsquigarrow^\omega_G y_c}
        \quad
        \inferrule*[Right={refl}]{ }{* \approx_G *}
      }{(b \wedge c)(y_c', *) \approx_G (b \wedge c)(y_c, *)}
      \qquad
      (b \wedge c)(x_b', *) = (b \wedge c)(y_c', *)
    }{(b \wedge c)(x_b, *) \approx_G (b \wedge c)(y_c, *)
    \qquad\qquad\qquad\qquad\qquad\qquad\;\; \forall b \wedge c \in P;Q}
  }{P;Q(\lambda b \wedge c. (b \wedge c)(x_b, *) \approx_G
      P;Q(\lambda b \wedge c. (b \wedge c)(y_c, *)}
    \end{mathpar}
    where $*$ is allowed to be any element of the $B_\cal{J}$-set
    $X$ (if $X$ is empty the theorem is vacuously true anyway).
    Here, the equality $(b \wedge c)(x_b', *) = (b \wedge
    c)(y_c', *)$ follows from $P(x') = Q(y')$, since
    $(b\wedge c)(P(x'), *) = c(b(P(x'), *), *) = c(b(b(P(x'),
    x'_b), *), *) = c(b(x'_b, *), *) = (b \wedge c)(x'_b, *)$ and
    similarly $(b\wedge c)(Q(y'), *) = (b \wedge c)(y'_c, *)$.

    Now, we see on the lefthand-side that $P(x) = P;Q(\lambda b
    \wedge c. x_b) = P;Q(\lambda b \wedge c. (b \wedge c)(x_b,
    *))$, and similarly for $Q(y)$ on the righthand-side. Each of
    the derivations of $(b \wedge c)(x_b, *) \approx_G (b \wedge
    c)(y_c, *)$ only uses finite congruences, so can be
    re-arranged into derivations of $(b \wedge c)(x_b, *)
    \leftrightsquigarrow^\omega_G (b \wedge c)(y_c, *)$, which
    concludes this inductive case and hence the proof.
  \end{proof}

  Now let $G$ be the generating equation of definition
  \ref{defn:sheaf-of-transitions}, and for which we will omit the
  subscript from this point on. The proof of the actual lemma
  proceeds in two steps. We first prove lemma
  \ref{lemma:pointwise-trace-equivalence-for-finitary-J}, which
  is the version of lemma \ref{lemma:pointwise-trace-equivalence}
  for $\rightsquigarrow^\omega_G$ (which actually suffices to
    prove lemma \ref{lemma:pointwise-trace-equivalence} in the case
  where $\cal{J}$ is finitary), and then prove the statement for $\approx$.

  \begin{lemma} \label{lemma:pointwise-trace-equivalence-for-finitary-J}
    For any $x, y \in T1[B]^{\cal{J}}$, if $x
    \leftrightsquigarrow^\omega y$ then for each $m, n \in T1$,
    we have $m \sim_{x(m) \wedge y(n)} n$.
  \end{lemma}
  \begin{proof}
    A derivation of $x \leftrightsquigarrow^\omega y$ is a
    (composable) chain of either $\rightsquigarrow$ or
    $\leftsquigarrow$ derivations, e.g. $x = x_0 \rightsquigarrow
    x_1 \leftsquigarrow x_2 \rightsquigarrow x_3 \rightsquigarrow
    x_4 \leftsquigarrow \ldots \rightsquigarrow x_k = y$ (arrow
    directions non-indicative). Each such derivation $x_i
    \rightsquigarrow x_{i+1}$, by lemma
    \ref{lemma:partition-associated-to-rewrite}, can be rewritten
    with exactly one \textsc{cong} rule over an associated
    partition $R$. This means that the derivation looks like
    \begin{mathpar}
      \inferrule*{\vdots}{x_i = R(h) \rightsquigarrow R(h') = x_{i + 1}}
    \end{mathpar}
    where for a unique $b_0 \in R$, we have $h_{b_0} = t \bind u$
    and $h'_{b_0} = P^{(t)}(\lambda a.t \rskip u(a))$ and for $b
    \neq b_0 \in R$, we have $h_b = h'_b$. Associate to $x_i$ the
    partition $P^{\leftarrow}_i := R$, and to $x_{i+1}$ the partition
    \[P^{\rightarrow}_{i + 1} := R \bind \lambda b.
      \begin{cases}
        P^{(t)} & \text{if } b = b_0 \\
        \set{\top} & \text{otherwise.}
    \end{cases} \]
    For a derivation $x_i \leftsquigarrow x_{i+1}$, perform the
    opposite assignment.
    So each $x_i$ is then associated with two partitions
    $P^{\rightarrow}_i$ and $P^{\leftarrow}_i$, except for $x_0$
    and $x_k$. For these, define $P^{\rightarrow}_0 := P$ and
    $P^{\leftarrow}_k := Q$. Since there are finitely many of
    these partitions, we can take a common refinement---call this $S$.

    Consider $d \in S$. It refines a unique $b \in P$, which
    identifies the term $t_d^0 := t_b$. Now look at the first
    derivation, and suppose it is $x_0 \rightsquigarrow x_1$.
    Then $d$ refines a unique $c \in P_1^\rightarrow$. There are
    two possible cases:
    \begin{enumerate}
      \item Either $c \in P^\leftarrow_0$, in which case we
        define $t_d^1 := t_d^0$;
      \item or $c = c_0 \wedge [t \mapsto a]$ for some $c_0 \in
        P^\leftarrow_0$, in which case we know that $t_d^0$ must
        be of the form $t \bind u$. So define $t_d^1 = t \rskip u(a)$.
    \end{enumerate}
    We note that in either case, we have $t_d^0 \sim_d t_d^1$.

    The other possibility is that $x_0 \leftsquigarrow x_1$. Then
    $d$ refines a unique $c \in P_0^\leftarrow$. There are two
    possible cases:
    \begin{enumerate}
      \item Either $c \in P^\rightarrow_1$, in which case we
        define $t_d^1 := t_d^0$;
      \item or $c = c_0 \wedge [t \mapsto a]$ for some $c_0 \in
        P^\rightarrow_1$, in which case we know that $t_d^0$ must
        be of the form $t \rskip u(a)$. So define $t_d^1 = t \bind u$.
    \end{enumerate}
    Now we may repeat this process, obtaining $t_b = t_d^0 \sim_d
    t_d^1 \sim_d \ldots \sim_d t_d^k$. Here, since $d$ refines
    some $c \in Q$, we have that $t_d^k = s_c$. So we may
    conclude $t_b \sim_d s_c$. To finish the proof, we see that
    any $b \wedge c \in P \wedge Q$ is a join of its refinements
    in $S$, and since we show $t_b \sim_d s_c$ for all of its
    refinements $d$, we can conclude that $t_b \sim_{b \wedge c} s_c$.
  \end{proof}

  Finally, we prove lemma \ref{lemma:pointwise-trace-equivalence}.

  $(\implies)$ suppose $x_1 \approx x_2 \in T1[B]^{\cal{J}}$.
  Then by lemma \ref{lemma:normal-form-of-approxG}, we know that
  $x_1 = P(x) \approx P(y) = x_2$ for some partition $P \in
  \cal{J}$ and families $x, y \colon P \to T1[B]^{\cal{J}}$ such
  that for each $b \in P$, $x_b \leftrightsquigarrow^\omega y_b$.
  So by lemma
  \ref{lemma:pointwise-trace-equivalence-for-finitary-J}, we have
  $m \sim_{x_b(m) \wedge y_b(n)} n$ for every $m, n \in T1$. Now,
  recall from \ref{defn:free-BJ-sets} that $P(x)(m) := \bigvee_{b
  \in P} b \wedge x_b(m)$, so $P(x)(m) \wedge P(y)(n) =
  \bigvee_{b \in P} b \wedge x_b(m) \wedge y_b(n)$. Hence
  $P(x)(m) \wedge P(y)(n) \leq \deno{m \sim n}$ iff for all $b
  \in P$, $b \wedge x_b(m) \wedge y_b(n) \leq \deno{m \sim n}$,
  which we have.

  $(\impliedby)$ Suppose $m \sim_{x_1(m) \wedge x_2(n)} n$ for
  each $m, n \in T1$.
  Let $P = \set{x_1(m) \wedge x_2(n) | m, n \in T1}^-$ be the
  common refinement
  of $x_1$ and $x_2$. Abusing notation, we will write families indexed by
  elements of $P$ as indexed by pairs $(m, n)$, for example
  $\lambda (m, n).x_{(m, n)}$.
  Then $x_1 = P(\lambda (m, n).
  m)$ and $x_2 = P(\lambda (m, n). n)$, and so by
  \textsc{cong-$P$} it suffices to prove
  $b(m, n) \approx n$ for each $m, n \in T1$ and $b \in P$ with
  $m \sim_b n$.

  Consider first the special case where $b \leq \deno{m \sim_1 n}$. Then
  \[ \left\{ [t \mapsto a] ~\middle|~
      \begin{matrix}
        A \in \Set, |A| \leq \kappa, t : TA, u, v \colon A \to T1,
        a \in A, \\
        u(a) = v(a), m = t \bind u, n = t \bind v
  \end{matrix}\right\} \cup \set{\neg b} \]
  is an open cover, so there is a refining partition $P$. We can
  further refine
  this partition to $Q = P; \set{b, \neg b}$. Now each $q \in Q$
  is either $q \leq \neg b$, or $q \leq b$ and associated with
  some $t_q \in TA$, $a_q \in A$ and families $u_q, v_q$ such
  that $q \leq [t_q \mapsto a_q]$, $u_q(a) = v_q(a)$, $m = t_q
  \bind u_q$ and $n = t_q \bind v_q$. Then we can derive
  \begin{alignat*}{2}
    b(m, n) &= Q \left(\lambda q. \left.
      \begin{cases}
        m & q \leq b \\
        n & q \leq \neg b \\
  \end{cases}\right\}\right) \\
  &= Q \left(\lambda q. \left.
    \begin{cases}
      t_q \bind u_q & q \leq b \\
      n & q \leq \neg b \\
\end{cases}\right\}\right) \\
&\approx Q \left(\lambda q. \left.
  \begin{cases}
    P^{(t_q)}(\lambda a. t_q \rskip u_q(a)) & q \leq b \\
    n & q \leq \neg b \\
\end{cases}\right\}\right) \qquad && \text{ by definition of } \approx \\
&= Q \left(\lambda q. \left.
\begin{cases}
  t_q \rskip u_q(a_q) & q \leq b \\
  n & q \leq \neg b \\
\end{cases}\right\}\right) \qquad && \text{ since } q \leq [t_q \mapsto a_q] \\
&= Q \left(\lambda q. \left.
\begin{cases}
t_q \rskip v_q(a_q) & q \leq b \\
n & q \leq \neg b \\
\end{cases}\right\}\right) \\
&= b(n, n) = n \qquad && \text{ by similar reasoning}
\end{alignat*}

Now, consider the general case: by ultraparacompactness, it suffices
to consider when $b \leq \deno{m \sim_k n}$ for each $k \in \bb{N}$.
Then $\set{\textstyle{\bigwedge^{k-1}_{i = 1}} \deno{m_i \sim_1
m_{i+1}} | m_1 = m, m_2 \ldots m_{k - 1} \in T1, m_k = n} \cup \set{\neg b}$
is refinable by a partition $Q$ such that each $q \in P$ is either $q
\leq \neg b$ or $q \leq b$ and $q \leq  \textstyle{\bigwedge^{k-1}_{i
= 1}} \deno{m_i \sim_1 m_{i + 1}}$ for some $\set{m_i}_{i \leq k}$.
Then we can prove $q(m_i, n) \approx q(m_{i+1}, n)$ by similar
reasoning as the previous paragraph, so by transitivity of $\approx$
we have $q(m, n) \approx q(n, n) = n$. Then we finally finish the proof with
the following equational reasoning:
\begin{alignat*}{2}
b(m, n) &= Q \left(\lambda q. \left.
\begin{cases}
m & q \leq b \\
n & q \leq \neg b \\
\end{cases}\right\}\right) \\
&= Q \left(\lambda q. \left.
\begin{cases}
q(m, n) & q \leq b \\
n & q \leq \neg b \\
\end{cases}\right\}\right) \\
&\approx Q \left(\lambda q. \left.
\begin{cases}
q(n, n) & q \leq b \\
n & q \leq \neg b \\
\end{cases}\right\}\right) \\
& = b(n, n) = n.
\end{alignat*}

\subsection{Proof that definition
\ref{defn:locale-of-transitions} corresponds to definition
\ref{defn:sheaf-of-transitions}}
\label{subsec:proof-of-locale-of-transitions}

Suppose $w \colon T1 \to \nm{LB}_0T$ is a function respecting
trace equivalence as in \ref{defn:locale-of-transitions}. Then
we need to show $w(t \bind u) = P^{(t)}(\lambda [t \mapsto a].
t \rskip u(a))$, which we have by
\begin{alignat*}{2}
w(t \bind u) & = P^{(t)}(\lambda [t \mapsto a]. w(t \bind u))
&& \qquad \text{\eqref{eqn:BJ-set-axioms}} \\
& = P^{(t)}(\lambda [t \mapsto a]. [t \mapsto a](w(t \bind
u), w(t \rskip u(a))))
&& \qquad \text{\eqref{eqn:BJ-set-axioms}} \\
& = P^{(t)}(\lambda [t \mapsto a]. w(t \rskip u(a)))
&& \qquad \text{(*)}
\end{alignat*}
where (*) follows because $w$ respects trace equivalence:
\begin{alignat*}{2}
& t \bind u \sim_{[t \mapsto a]} t \rskip u(a) \\
\implies & w(t \bind u) \equiv_{[t \mapsto a]} w(t \rskip u(a)) \\
\iff  & [t \mapsto a](w(t \bind u), w(t \rskip u(a))) = w(t \rskip u(a))
\end{alignat*}

On the other hand, if $\tilde{w} \colon F_T \to \nm{LB}_0T$ is
a $B_{\cal{J}}$-set homomorphism, then we need to show
$\tilde{w}$ restricts to a function $w \colon T1 \to
\nm{LB}_0T$ which respects trace equivalence. Consider then two
trace equivalent terms $m \sim_b n \in T1$. Then we have
$\tilde{w}(m) \equiv_b \tilde{w}_n$ since
\[ b(\tilde{w}(m), \tilde{w}(n)) = \tilde{w}(b(m, n)) = \tilde{w}(n) \]
where the last equality follows because $b(m, n) = n \iff m
\sim_b n$ by lemma \ref{lemma:pointwise-trace-equivalence}.

\subsection{Proof of proposition \ref{prop:germs-of-LB1T}}

By proposition \ref{prop:germs-of-local-homeo}, we know that
$\pt(\nm{LB}_1T) \cong \Sigma_{\beta \in \pt\nm{LB}_0T}
F_T/_{\equiv_\beta}$. But we know $\pt\nm{LB}_0T \cong
\bb{B}_0T$, so the $\beta$ really are just admissible behaviours.
Next, observe that every $x \in F_T$ can be expressed in the form
$x = P(\lambda b. m_b)$, and hence $x \equiv_b m_b$ for the $b
\in P$ with $\beta \in b$. Hence, we have $F_T/_{\equiv_b} \cong
T1/_{\equiv_\beta} \cong T1/_{\sim_\beta}$ over each $\beta$.
Therefore $\pt(\nm{LB}_1T) \cong \Sigma_{\beta \in \bb{B}_0T}
T_1/_{\sim_\beta} = \bb{B}_1T$.

\subsection{Proof of lemma \ref{lemma:pullback-with-LB1T}}

The following decomposition lemma, analogous to lemma
\ref{lemma:basis-of-LB1T}, will come in handy.

\begin{lemma} \label{lemma:basis-of-pullback-with-LB1T}
Every $h \colon T1 \to L$ with the conditions of this lemma
is of the form $h = \bigvee_{m} m^* \wedge \nm{const}_{h(m)}$
where $m^* := \lambda n. f\inv\deno{m \sim n}$.
\end{lemma}
\begin{proof}
We have to show $h(n) = \bigvee_{m} f\inv\deno{m \sim n}
\wedge h(m)$. As in the proof of lemma
\ref{lemma:basis-of-LB1T}, the left-to-right inequality is
easy, so we focus on the right-to-left inequality for which
we have to show $f\inv\deno{m \sim n} \wedge h(m) \leq h(n)$.
By ultraparacompactness, $f\inv\deno{m \sim n} = \bigvee
\set{f\inv b | b \leq \deno{m \sim n}, b \in B}$ so it
suffices to prove for each complemented $b \leq \deno{m \sim
n}$ that $f\inv b \wedge h(m) \leq h(n)$, but this
immediately follows from the condition on $h$.

It is easy to see that $\nm{const}_{h(m)}$ satisfies the
condition since it is just a constant map. Meanwhile, for
$m^*$ whenever $n_1 \sim_b n_2$ we have that \[ m^*n_1 \wedge
f\inv b = f\inv\deno{m \sim n_1} \wedge f\inv b =
f\inv(\deno{m \sim n_1} \wedge b) = f\inv(\deno{m \sim n_2}
\wedge b) m^*n_2 \wedge f^*b. \]
Hence $m^*$ satisfies the condition of this lemma.
\end{proof}

The two projections $\pi_1 \colon L \times_{\nm{LB}_0T}
\nm{LB}_1T \to L$ and $\pi_2 \colon L \times_{\nm{LB}_0T}
\nm{LB}_1T \to \nm{LB}_1T$ are given by $\pi\inv_1 \colon u
\mapsto \nm{const}_u$ and $\pi\inv_2 \colon w \mapsto \lambda
m. f\inv w(m)$. Given a pullback cone $i \colon Z \to L$ and $j
\colon Z \to \nm{LB}_1T$, if a universal arrow $\braket{i, j}$
exists then it must satisfy
\[ \braket{i, j}\inv(\nm{const}_{h(m)}) = \braket{i,
j}\inv\pi\inv_1h(m) = i\inv h(m) \]
\[ \braket{i, j}\inv(m^*) = \braket{i,
j}\inv\pi\inv_2h(\hat{m}) = j\inv \hat{m} \]
But then by lemma \ref{lemma:basis-of-pullback-with-LB1T}, the
frame of opens for the pullback is generated by these opens, so
these two equations uniquely determine $\braket{i, j}$. It is
then straightforward to check that this is well-defined.

\subsection{Proof that definition
\ref{defn:localic-behaviour-category} is a localic category}
\label{subsec:proof-of-localic-behaviour-category}

It is very straightforward to check that the domain and
codomain of identity and compositions correspond to what they
should be, so we focus on the unitality and associativity. It
is easy to see from lemma \ref{lemma:pullback-with-LB1T} that
$\nm{LB}_0T \times_{\nm{LB}_0T} \nm{LB}_1T \cong \nm{LB}_1T$,
for which the unitality diagram becomes
\[
\begin{tikzcd}[ampersand replacement=\&,cramped]
{\nm{LB}_1T} \&\& {\nm{LB}_1T \times_{\nm{LB}_0T}
\nm{LB}_1T} \&\& {\nm{LB}_1T} \\
\\
\&\& {\nm{LB}_1T}
\arrow["{h \mapsto h(\return, -)}", from=1-1, to=1-3]
\arrow[equals, from=1-1, to=3-3]
\arrow["{w \mapsto w(- \rskip -)}", from=1-3, to=3-3]
\arrow["{h \mapsto h(-, \return)}"', from=1-5, to=1-3]
\arrow[equals, from=3-3, to=1-5]
\end{tikzcd}
\]
and this commutes by unitality of $\rskip$. Finally, by lemma
\ref{lemma:pullback-with-LB1T} the pullback of composable
triples $\nm{LB}_1T \times_{\nm{LB}_0T} \nm{LB}_1T
\times_{\nm{LB}_0T} \nm{LB}_1T$ can be constructed as the frame
of appropriate maps $T1 \times T1 \times T1 \to
\cal{O}\nm{LB}_0T$, for which the associativity diagram below
obviously commutes due to the associativity of $\rskip$.
\[
\begin{tikzcd}[ampersand replacement=\&,cramped]
{\nm{LB}_1T \times_{\nm{LB}_0T} \nm{LB}_1T
\times_{\nm{LB}_0T} \nm{LB}_1T} \&\& {\nm{LB}_1T
\times_{\nm{LB}_0T} \nm{LB}_1T} \\
\\
{\nm{LB}_1T \times_{\nm{LB}_0T} \nm{LB}_1T} \&\& {\nm{LB}_1T}
\arrow["{h \mapsto h(-,- \rskip -)}", from=1-1, to=1-3]
\arrow["{h \mapsto h(- \rskip -, -)}"', from=1-1, to=3-1]
\arrow["{w \mapsto w(- \rskip -)}", from=1-3, to=3-3]
\arrow["{w \mapsto w(- \rskip -)}"', from=3-1, to=3-3]
\end{tikzcd}\]

\subsection{Proof of lemma \ref{lemma:global-sections-of-B1T-is-F_T}}

Each $\beta \in \bb{B}_0T$ admits an open neighborhood
$s\inv[m_\beta | \return 1]$ where $t_\beta$ is some
representative of the equivalence class $s(\beta)$. This induces
an open cover on $\bb{B}_0T$ which is refined by a finite
partition $\set{b_i}_{i \in I}$ since $\bb{B}_0T$ is a Stone
space. It suffices to pick $m_i$ to be the $m_\beta$ of an open
$s\inv[m_\beta | \return 1]$ refined by $b_i$. The uniqueness
under trace equivalence is easy to see because for any $\beta \in
b_i \wedge b_j$, we have $[m_i]_\beta = s(\beta) = [m_j]_\beta$.
The spatiality of $\nm{LB}_0T$ then ensures this corresponds to
$m_i \sim_{b_i \wedge b_j} m_j$.

\subsection{Details to the proof of proposition
\ref{prop:functoriality-of-LB}}
\label{subsec:proof-of-functoriality-of-LB}

First we have to verify that $\nm{LB}\varphi$ is an internal
retrofunctor. For this we need to consider the pullbacks
\[
\begin{tikzcd}[ampersand replacement=\&,cramped]
{\Lambda_1} \&\& {\nm{LB}_1T} \& {\Lambda_2} \&\& {\nm{LB}_2T} \\
\\
{\nm{LB}_0S} \&\& {\nm{LB}_0T} \& {\Lambda_1} \&\& {\nm{LB}_1T}
\arrow["{\pi_2}", from=1-1, to=1-3]
\arrow["{\pi_1}"{description}, from=1-1, to=3-1]
\arrow["\lrcorner"{anchor=center, pos=0.125}, draw=none,
from=1-1, to=3-3]
\arrow["\sigma", from=1-3, to=3-3]
\arrow["{\pi_2}"', from=1-4, to=1-6]
\arrow["{\pi_1}", from=1-4, to=3-4]
\arrow["\lrcorner"{anchor=center, pos=0.125}, draw=none,
from=1-4, to=3-6]
\arrow["{\pi_1}"', from=1-6, to=3-6]
\arrow["{\nm{LB}_0\varphi}"', from=3-1, to=3-3]
\arrow["{\pi_2}", from=3-4, to=3-6]
\end{tikzcd}\]
Noting that $\Lambda_2$ is the pullback of the second square composed with
the first square, and following similar ideas to lemma
\ref{lemma:pullback-with-LB1T}, the second pullback $\Lambda_2$
can be expressed
by the frame of maps $h : T1 \times T1 \to \cal{O}\nm{LB}_0S$ such that
$m \sim_b m'$ implies $h(m, n) \wedge (\nm{LB}_0\varphi)\inv b
= h(m, n) \wedge (\nm{LB}_0\varphi)\inv b$, and $n \sim_b n'$
implies $h(m, n) \wedge (\nm{LB}_0\varphi)\inv \codeno{m}\inv b
= h(m, n') \wedge (\nm{LB}_0\varphi)\inv \codeno{m}\inv b$.

The requirements on the domain and codomain of the lift is
encoded by requiring
the following diagram to commute:
\[
\begin{tikzcd}[ampersand replacement=\&,cramped]
\&\& {\Lambda_1} \&\& {\nm{LB}_1T} \\
\\
{\nm{LB}_0S} \&\& {\nm{LB}_1S} \& {\nm{LB}_0S} \& {\nm{LB}_0T}
\arrow["{\pi_2}", from=1-3, to=1-5]
\arrow["{\pi_1}"', from=1-3, to=3-1]
\arrow["{\nm{LB}_1\varphi}"', from=1-3, to=3-3]
\arrow["\tau", from=1-5, to=3-5]
\arrow["\sigma"', from=3-3, to=3-1]
\arrow["\tau"', from=3-3, to=3-4]
\arrow["{\nm{LB}_0\varphi}"', from=3-4, to=3-5]
\end{tikzcd}\]
but this follows by a straightforward chase along the diagram.
Next, to see that identity and composition is respected, we
require the following diagrams to commute:
\[
\begin{tikzcd}[ampersand replacement=\&,cramped]
{\Lambda_2} \&\& {\Lambda_1} \&\& {\nm{LB}_0S} \\
\\
{\nm{LB}_2S} \&\& {\nm{LB}_1S}
\arrow["{\id \times \mu}", from=1-1, to=1-3]
\arrow["d"', from=1-1, to=3-1]
\arrow["{\nm{LB}_1\varphi}"', from=1-3, to=3-3]
\arrow["{\braket{\id, \iota \circ \nm{LB}_0}}", from=1-5, to=1-3]
\arrow["\iota", from=1-5, to=3-3]
\arrow["\mu"', from=3-1, to=3-3]
\end{tikzcd}\]
For the commutativity involving $\iota$, this amounts to
checking, for $w \in \cal{O}\nm{LB}_1S$, that
\[\bigvee_{m \in T1} w(\varphi(m)) \wedge
(\nm{LB}_0\varphi)\inv \deno{m \sim \return} = w(\return). \]
The proof of this is similar to the proof of lemma
\ref{lemma:basis-of-pullback-with-LB1T}. For the commutativity
involving $\mu$, the map $d$ is defined on inverse image
by $d\inv : h \mapsto h(\varphi_1(-), \varphi_1)$. Then the
commutativity of this
square amounts to checking, for $w \in \cal{O}\nm{LB}_1T$, that
$w(\varphi(-) \rskip \varphi(-)) = w(\varphi(- \rskip -))$ but
this easily follows from $\varphi$ being
a monad map.

We now also must verify that $\nm{LB}$ is functorial.
If $\varphi = \id : T \to T$ then we see that the definition of
$\nm{LB}\varphi$ is indeed the identity retrofunctor. For
composition, given
$\varphi : T \to S$ and $\psi : S \to R$, it
is obvious that $\nm{LB}_0(\varphi \circ \psi) =
\nm{LB}_0(\psi) \circ \nm{LB}_0(\varphi)$. For the action on
morphisms, the composite is given by
\[
\begin{tikzcd}[ampersand replacement=\&,cramped]
{\nm{LB}_0R \times_{\nm{LB}_0T} \nm{LB}_1T} \&\&\&
{\nm{LB}_0R \times_{\nm{LB}_0S} \nm{LB}_1S} \&\& {\nm{LB}_1R}
\arrow["{\braket{\pi_0, \nm{LB}_1\varphi \circ
(\nm{LB}_0\psi \times \id)}}", from=1-1, to=1-4]
\arrow["{\nm{LB}_1\psi}", from=1-4, to=1-6]
\end{tikzcd}\]

Let us compute the inverse image of an open set $w \in
\cal{O}(\nm{LB}_1R)$ along this map. The inverse along
$\nm{LB}_1\psi$ gives
$w \circ \psi_1$ which by lemma
\ref{lemma:basis-of-pullback-with-LB1T} can be expressed as
$\bigvee_{s \in S1} s^* \wedge \nm{const}_{w(\psi(s))}$. Now
the inverse of $\nm{const}_{w(\psi(s))}$ along the
pair of maps can be computed as the inverse of just the left
component, which again gives
$\nm{const}_{w(\psi(s))}$, but this time as an open in
$\nm{LB}_0R \times_{\nm{LB}_0T} \nm{LB}_1T$. The inverse of
$s^*$ along the pair can be computed as the inverse of
$\hat{s}$ along the right component, which gives $\lambda t \in
T1. (\nm{LB}_0\psi)\inv\deno{s \sim \varphi(t)}$. So combining
the two, in the end we get an open of $\nm{LB}_0 R
\times_{\nm{LB}_0T} \nm{LB}_1T$ defined by
\[ \lambda t \in T1. \bigvee_{s \in S1} (\nm{LB}_0\psi)\inv
\deno{s \sim \varphi(t)} \wedge w(\psi(s)) \]
and we have to show this is equal to $(\nm{LB}_1(\psi \circ
\varphi))\inv w = w \circ \psi_1 \circ \phi_1$. But this is
again the same type of reasoning as in the proof of lemma
\ref{lemma:basis-of-pullback-with-LB1T}.

\subsection{Details to the proof of theorem
\ref{thm:stone-adjunction-loc}}
\label{subsec:proof-of-stone-adjunction-loc}
Let us write $\Gamma$ as shorthand for $\Gamma\nm{LC}$.
It suffices to prove, for any retrofunctor $F \colon \nm{LC} \to
\nm{LB}T$, there is a unique monad morphism $\varphi \colon T
\to \Gamma$ such that $\nm{LB}\varphi \circ \varepsilon = F$.
For this, we show that this condition uniquely determines
$\varphi$. So consider $t \in TA$ and  observe that
$\varphi(t)\inv \braket{a \mapsto w} = \varphi(t)\inv\braket{a
\mapsto \top} \wedge \varphi(t \rskip \return)\inv w$. We then show that
\begin{enumerate*}
\item $F_0\inv[t \mapsto a] = \varphi(t)\inv\braket{a \mapsto \top}$; and
\item $(F_1\inv w)(t \rskip \return) = \varphi(t \rskip \return)\inv w$.
\end{enumerate*}
This fully determines $\varphi(t)\inv \braket{a \mapsto w}$ as
$F_0\inv[t \mapsto a] \wedge (F_1\inv w)(t \rskip \return)$.

Now, (i) is a straightforward unfolding of definitions on the
equation $F_0\inv = \varepsilon_0\inv (\nm{LB}_0\varphi)\inv$,
so we leave this as an exercise to the reader (if the reader is
still reading). For (ii), we have
\[
\begin{tikzcd}[ampersand replacement=\&,cramped]
{F_1 = \nm{LC}_0 \times_{\nm{LB}_0T} \nm{LB}_1T} \&\&\&
{\nm{LC}_0 \times_{\nm{LB}_0\Gamma} \nm{LB}_1\Gamma} \& {\nm{LC}_1}
\arrow["{\braket{\pi_0, \nm{LB}_1\varphi \circ
(\varepsilon_0 \times \id)}}", from=1-1, to=1-4]
\arrow["{\varepsilon_1}", from=1-4, to=1-5]
\end{tikzcd}\]
Let us compute the inverse image of $w \in \nm{LC}_1$ along
this map. First, by lemma
\ref{lemma:basis-of-pullback-with-LB1T} we can decompose
$\varepsilon_1\inv w = \bigvee_{m \in \Gamma 1}
\nm{const}_{m\inv w} \wedge m^*$. Then the the inverse image of
$\nm{const}_{m\inv w}$ along the pair
is given simply as $\nm{const}_{m\inv w}$ while the inverse
image of $m^*$ is
$\lambda n \in T1. \varepsilon_0\inv \deno{\varphi(n) \sim m}$.
Therefore, we arrive at the result
\[ F_1\inv w = \bigvee_{m} \lambda n \in T1. m\inv w \wedge
\varepsilon_0\inv \deno{\varphi(n) \sim m} \]
Hence, if we let $n := t \rskip
\return$, then we have $F_1\inv w(n) = \bigvee_{m \in \Gamma1}
m\inv w \wedge
\varepsilon\inv\deno{\varphi(n) \sim m}$. Now, it is easy to
see that $\varphi(n)\inv w \leq
F_1\inv w(n)$ by taking $m := \varphi(n)$. On the other hand,
for $\varphi(n)\inv w \geq F_1\inv w(n)$ we have to reason in
terms of witnesses of $\deno{\varphi(n) \sim m}$. For
simplicity, we simply consider a $1$-step witness $[h \mapsto
b]$ for $h \in \Gamma B, b \in B$ such that $\varphi(n) = h
\bind u$ and $m = h \bind v$, with $u(b) = v(b)$. Then one can see that
\begin{alignat*}{2}
m\inv w \wedge \varepsilon_0\inv [h \mapsto b]
&= (h \bind v)\inv w \wedge h\inv\braket{b \mapsto \top} \\
&= (h \rskip v(b))\inv w \wedge h\inv\braket{b \mapsto \top} \\
&= (h \bind u)\inv w \wedge h \inv\braket{b \mapsto \top} \\
&\leq \varphi(n)\inv w
\end{alignat*}
This proves (ii) and hence we conclude that $\varphi$ is
uniquely determined.

\subsection{Proof of proposition
\ref{prop:global-sections-monad-is-haffun}}
\label{subsec:proof-of-global-sections-monad-is-haffun}

\begin{proof}
A computation reveals that $(h \rskip \return)\inv w =
\bigvee_{a} h\inv\braket{a \mapsto w}$, while $\return\inv w =
\iota\inv w$. So these two are equal iff $h \rskip \return a =
\return a$. Notice that we do not use the second condition of
being hyperaffine, because it is automatically true for any $h$
satisfying the first condition (also known as \emph{affine}).
Now, suppose that such a hyperaffine $\overline{s}$ for a section
$s \in \Gamma\nm{LC} (A)$. Then a straightforward computation
(using the characterization of hyperaffines) reveals that the
condition $s = \overline{s} \bind \lambda a. s \rskip \return a$
implies $\overline{s}\inv \braket{a \mapsto w} = s\inv\braket{a
\mapsto \top} \wedge \iota\inv w$ which determines
$\overline{s}\inv$ as a well-defined frame homomorphism. See
below for the computations.
\end{proof}

(affine characterization) Applying the definition of $\bind$,
we find that $(h \rskip \return)\inv w = \bigvee_a
h\inv\braket{a \mapsto w'}$ where $w' = \bigvee \set{v_1 \wedge
\tau\inv\iota\inv v_2 | v_1 \times v_2 \leq \mu\inv w}$. But
notice that $w'$ is the inverse image of $w$ along
\[
\begin{tikzcd}[ampersand replacement=\&,cramped]
{\nm{LC}_1} \&\& {\nm{LC}_1 \times_{\nm{LC}_0} \nm{LC}_0}
\& {\nm{LC}_1 \times_{\nm{LC}_0} \nm{LC}_0} \& {\nm{LC}_1}
\arrow["{\braket{\id, \tau}}", from=1-1, to=1-3]
\arrow["{\id \times \iota}", from=1-3, to=1-4]
\arrow["{\pi_{\nm{LC}_1}}"', curve={height=18pt}, from=1-3, to=1-5]
\arrow["\mu", from=1-4, to=1-5]
\end{tikzcd}\]
but this inverse image is equally well computed as $w \wedge
\tau\inv\top = w$, and hence $w' = w$.

(determination of $\overline{s}$) Since $s = \overline{s}\inv
\bind \lambda a. s \rskip \return a$, we can compute
$s\inv\braket{a \mapsto \top}$ as
\[\overline{s}\inv\Braket{a \mapsto \bigvee \set{v_1 \wedge
\tau\inv s\inv \bigvee_{a' \in A} | v_1, v_2 \in \cal{O}\nm{LC}_1 }} \]
Next, we know $\iota\inv w = \bigvee_{a'' \in A}
\overline{s}\inv\braket{a'' \mapsto w}$ so
\[\iota\inv w \wedge s\inv\braket{a \mapsto \top} =
\overline{s}\inv\Braket{a \mapsto w \wedge \bigvee \set{ v_1
\wedge \tau\inv s\inv \bigvee_{a' \in A}\braket{a' \mapsto v_2}
| v_1, v_2 \in \cal{O}\nm{LC}_1 }}\]
and now we have $w \leq \bigvee \set{ v_1 \wedge \tau\inv s\inv
\bigvee_{a' \in A}\braket{a' \mapsto v_2} | v_1, v_2 \in
\cal{O}\nm{LC}_1 }$ by taking $v_1 = w$ and $v_2 = \top$, so
this simplifies to $\overline{s}\inv \braket{a \mapsto w}$.

\subsection{Proof of lemma \ref{lemma:H2-of-hyperaffine-unary-monad}}

The Grothendieck Boolean algebra structure $H2_{\cal{J}}$ is
established in \cite{Garner2024-yc}. The Boolean algebra
structure on $H2$ is given by $\top = \return 1$, $h_1 \wedge h_2
= h_1 \bind (0 \mapsto \return 0; 1 \mapsto h_2)$ and $\neg h = h
\bind (0 \mapsto \return 1; 1 \mapsto \return 0)$, and from this
it is easy to see that $H2$ satisfies \eqref{eqn:LB-bot},
\eqref{eqn:LB-return}, $[t \rskip \return a \mapsto a']$, and
$\bigvee_{a \in A} [t \mapsto a] = \top$ for finite sets $A$.
Then the only missing axioms are $\bigvee_{a \in A} [t \mapsto a]
= \top$ for infinite $A$, but these are precisely the partitions
in $\cal{J}$. Hence $H2_{\cal{J}}$ generates $\cal{O}(\nm{LB}_0T)$.

The inverse to $\delta \colon T_1 \to F_T$ is witnessed by
$\delta\inv \colon x \mapsto h \bind \lambda b. x\inv b$, where
$h \in HP$ is a hyperaffine realizing the partition $P^{(h)}$
induced by $x \colon T1 \to H2$. On one hand we have
$\delta\delta\inv(x) = \delta(h \bind \lambda b. x\inv b) =
P(\delta(h \rskip x\inv b)) = P(\delta(x\inv b)) = x$. On the
other hand, $\delta\inv\delta(t) = \return \top \bind \lambda b.
\delta(t)\inv b = \delta(t)\inv \top = t$.

\subsection{Proof of proposition
\ref{prop:left-fixpoints-are-haffun}}
\label{subsec:proof-of-left-fixpoints-are-haffun}

($\impliedby$) Suppose now the unit map is an isomorphism. Then
the hyperaffine-unary factorization of $\Gamma\nm{LB}T$
(proposition \ref{prop:global-sections-monad-is-haffun}) must
transfer along the unit map onto $T$.

($\implies$) For the converse direction, we make use of lemma
\ref{lemma:H2-of-hyperaffine-unary-monad}, which basically says
any global section $s \in \Gamma\nm{LB}T A$ identifies a
hyperaffine $h \in HA$
and a family $u : A \to T1$ of unary computations, and the composite
$h \bind \lambda a. u(a) \rskip \return a$ induces the section
$s$. What follows is the proof in more detail.

Assume $T$ is
hyperaffine-unary, we have to prove $\eta_T$ is an isomorphism,
i.e., bijective at each level. To see that $\eta_T$ is surjective,
consider then a section $s \in \Gamma\nm{LB}TA$. We can always
factor $s = h \bind \lambda i. \eta(m_i) \rskip \return
f(i)$ for some hyperaffine section $h \in HC$,
some family of $T1$-terms $\set{m_i}_{i \in I}$ and function $f
\colon I \to A$. So it suffices to show that $h$ is in the image
of $\eta_T$. The data of a hyperaffine section $h$ is completely
determined by the partition $\set{h\inv\braket{i \mapsto \top} |
i \in I}^- \subseteq \nm{LB}_0T$, but now because $T$ is
hyperaffine-unary, by lemma
\ref{lemma:H2-of-hyperaffine-unary-monad} such a partition has to
be of the form $\set{[h' \mapsto j] | j \in J }$ for some $h' \in
HJ \subseteq TJ$ and $J = \set{i \in I | h\inv\braket{i \mapsto
\top} \neq \bot} \subseteq I$. We claim that $\eta_T(h') = h$,
and this follows by unfolding definitions:
\[\eta_T(h')\inv\braket{i \mapsto w} = [h' \mapsto i] \wedge w(h'
\rskip \return) = h\inv\braket{i \mapsto \top} \wedge w(\return)
= h\inv\braket{i \mapsto w} \]
Finally, to see that $\eta_T$ is injective, consider $t_1 \neq
t_2 \in TA$. Then again by proposition
\ref{prop:global-sections-monad-is-haffun}, both admit
decompositions $t_1 = \overline{t_1} \bind \lambda a. m_1 \rskip
\return a$ and $t_2 = \overline{t_2} \bind \lambda a. m_2 \rskip
\return a$, so if they are not equal it must be that either
$\overline{t_1} \neq \overline{t_2} \in HA$ or $m_1 \neq m_2 \in
T1$. If the former, then by hyperaffineness of $h_1$ and $h_2$:
\begin{alignat*}{2}
h_1 &= h_1 \bind \lambda a. h_1 \bind \lambda a'. \return a
\text{ if } a = a' \text{ else } h_2 & \text{($h_1$ is h.aff.)} \\
&= h_1 \bind \lambda a. [h_1 \mapsto a] \bind (\return a, h_2)
& \text{(definition of $[h_1 \mapsto a]$)} \\
&= h_1 \bind \lambda a. [h_2 \mapsto a] \bind (h_2 \rskip
\return a, h_2) & \text{($[h_1 \mapsto a] = [h_2 \mapsto a]$
and $h_2$ is h.aff.)} \\
&= h_1 \bind \lambda a. h_2 \bind \lambda a'. h_2 \bind \lambda
a''. \return a \text{ if } a = a'' \text{ else } a'' &
\text{(definition of $[h_2 \mapsto a]$)} \\
&= h_1 \bind \lambda a. h_2 \bind \lambda a'. h_2 \bind \lambda
a''. \return a'' \text{ if } a' = a'' \text{ else } \ldots
\qquad & \text{(the $\ldots$ does not matter)} \\
&= (h_1 \bind \lambda a. h_2) = (h_1 \rskip h_2) = h_2 &
\text{($h_2$ is h.aff.)}
\end{alignat*}
If the latter is true, then $\eta(m_1), \eta(m_2) \colon
\nm{LB}_0T \to \nm{LB}_1T$, viewed as maps of local
homeomorphisms, correspond to maps of sheaves $\delta(m_1),
\delta(m_2) \colon 1 \to F_T$, and from the isomorphism $F_T
\cong T1$ of lemma \ref{lemma:H2-of-hyperaffine-unary-monad} we
know these cannot be equal. It can then be verified that if we
have two hyperaffines $h_1 \neq h_2 \in \Gamma\nm{LB}TA$ or unary
sections $s_1 \neq s_2 \in \Gamma\nm{LB}1$ then $h_1 \bind
\lambda a. s_1 \return a \neq h_2 \bind \lambda a. s_2 \return
a$, and hence $\eta(t_1) \neq \eta(t_2)$.

\subsection{Proof of proposition
\ref{prop:right-fixpoints-are-ample}}
\label{subsec:proof-of-right-fixpoints-are-ample}

The following lemma come in handy.

\begin{lemma} \label{lemma:retrofunctorial-iso-iff-functorial-iso}
Two internal categories are retrofunctorially isomorphic iff they
are functorially isomorphic.
\end{lemma}
\begin{proof}
Let $\nm{LC}$ and $\nm{LD}$ be internal categories. Given a
functorial isomorphism $F \colon \nm{LC} \to \nm{LD}$ with
inverse $F\inv$, define the retrofunctor $G$ by $G_0 := F_0$ and
$G_1 := F\inv_1 \circ \pi \colon \nm{LC}_0 \times_{\nm{LD}_0}
\nm{LD}_1 \to \nm{LD}_1 \to \nm{LC}_1$, and vice versa for
$G\inv$. On the other hand, given retrofunctors $G$ and $G\inv$,
define the functor $F$ by $F_0 := G_0$ and $F_1 := G\inv_1
\circ\braket{G_0 \sigma, \nm{id}}$ where $\braket{G_0 \sigma,
\nm{id}} \colon \nm{LC}_1 \to \nm{LD}_0 \times_{\nm{LC}_0}
\nm{LC}_1$. Define the inverse $F\inv$ similarly. We leave it to
the reader to verify the necessary equations.
\end{proof}

For brevity, we omit the subscript $\nm{LC}$ from $\varepsilon$,
and let us also write $\nm{LB}_i := \nm{LB}_i\Gamma\nm{LC}$ for
$i = 0, 1$.

$(\implies)$ By lemma
\ref{lemma:retrofunctorial-iso-iff-functorial-iso}, we get an
isomorphism $\nm{LB}_0\Gamma\nm{LC} \cong \nm{LC}_0$ so
$\nm{LC}_0$ is also ultraparacompact, and also we get an
isomorphism $\nm{LB}_1\Gamma\nm{LC} \cong \nm{LC}_1$ commuting
with the source maps, so the source map of $\nm{LC}$ is also a
local homeomorphism.

$(\impliedby)$ By lemma
\ref{lemma:retrofunctorial-iso-iff-functorial-iso} it suffices to
prove that the counit $\varepsilon$ partakes in a functorial
isomorphism. The action on objects $\varepsilon_0 \colon
\nm{LC}_0 \to \nm{LB}_0$ has inverse $\revepsilon_0$ given on
generating clopens $b \in \frk{B}\nm{LC}_0$ by $\revepsilon_0\inv
\colon b \mapsto [b^+]$ where $b^+ \colon \nm{LC}_0 \to 2 \cdot
\nm{LC}_1$ given by $(b^+)\inv \colon \braket{1 \mapsto w}
\mapsto b \wedge \iota_{\nm{LC}}\inv w$ and $(b^+)\inv \colon
\braket{0 \mapsto w} \mapsto \neg b \wedge \iota_{\nm{LC}}\inv
w$. This map is well-defined because it realizes all partitions
of $\nm{LC}_0$: any partition $P$ manifests as a section $P^+ \in
\Gamma\nm{LC}(P)$ defined analogously to $b^+$, and hence we have $\top =
\bigvee_{b \in P} [P^+ \mapsto b] = \bigvee_{b \in P} [b^+]$. It
is straightforward to see that $\varepsilon\inv_0
\revepsilon\inv_0 = \id$. On the other hand, to see that
$\revepsilon\inv_0 \varepsilon\inv_0 = \id$, consider a
generating open $[s]$ where $s \in \Gamma\nm{LC}2$. By
proposition \ref{prop:global-sections-monad-is-haffun} we have
its corresponding hyperaffine $\overline{s}$, and it is easy to
see that $[\overline{s}] = [s]$. Then, it is a matter of checking
that $(\varepsilon\inv_0[s])^+ = \overline{s}$.

This gives us an internal functor $\cal{E}$ with $\cal{E}_0 :=
\revepsilon_0$ and $\cal{E}_1 := %
\begin{tikzcd}[ampersand replacement=\&,cramped]
{\nm{LB}_1} \&\& {\nm{LC}_0 \times_{\nm{LB}_0} \nm{LB}_1} \& {\nm{LC}_1}
\arrow["{\braket{\revepsilon_0 \sigma, \id}}", from=1-1, to=1-3]
\arrow["{\varepsilon_1}", from=1-3, to=1-4]
\end{tikzcd}$
which more explicitly can be computed as $\cal{E}_1\inv w =
\lambda m. \revepsilon_0\inv m\inv w$. By proposition
\ref{prop:global-sections-monad-is-haffun} and lemma
\ref{lemma:H2-of-hyperaffine-unary-monad}, the local homeomorphism
$\sigma_{\nm{LB}}$ is induced by the $B_\cal{J}$-set
$\Gamma\nm{LC}1$, but this just corresponds to the sheaf
induced by the local homeomorphism $\sigma_{\nm{LC}}$, so we must
have $\sigma_{\nm{LC}} \cong \sigma_{\nm{LB}}$. The map $\cal{E}_1$
is the canonical map witnessing this isomorphism, up to a change of
base along the isomorphism $\revepsilon_0$. We leave it
to the reader to verify functoriality.

\subsection{Proof of theorem \ref{thm:stone-duality}}

It follows from proposition \ref{prop:right-fixpoints-are-ample} that
$\varepsilon_{\nm{LB}}$ is an isomorphism, since $\nm{LB}T$ is
ample for any monad $T$. Hence adjunction
\ref{thm:stone-adjunction-loc} is idempotent, and propositions
\ref{prop:left-fixpoints-are-haffun} and
\ref{prop:right-fixpoints-are-ample} characterize the fixpoints.
For the finitary monad case, we know that $\cal{O} : \Top \to
\Loc$ preserves
pullbacks along local homeomorphisms, which means it preserves
ample topological
categories and their retrofunctors. From this, we get a functor
$\cal{O} : \AmpTopRetro \to \AmpLocRetro$, for which the equivalence
$\nm{HUMnd}_r \simeq \AmpLocRetro$, when restricted to finitary
monads, factors through. This factorization
$\nm{HUMnd}_\omega \to \AmpTopRetro$ is essentially surjective on
objects, because now taking the global sections monad on an ample
topological category $\bb{C}$, the compactness of the base space
$\bb{C}_0$ ensures this monad is the monad $\Gamma_\omega\bb{C}$
of \emph{finitary} sections. Hence, we get an equivalence
$\nm{HUMnd}_\omega \simeq \AmpTopRetro$.

\end{document}